\documentclass[journal,onecolumn,draftclsnofoot,12pt]{IEEEtran}
\IEEEoverridecommandlockouts

\usepackage{multicol}
\usepackage{amssymb}
\usepackage{amsmath}
\usepackage{amsfonts}
\usepackage{amsthm}
\usepackage{algorithm,algpseudocode}
\usepackage{bm}
\usepackage{caption}
\usepackage[noadjust]{cite}
\usepackage{chemarrow}
\usepackage{graphicx}
\usepackage{textcomp}
\usepackage{xcolor}
\usepackage{siunitx}
\usepackage[nolist,nohyperlinks]{acronym}
\usepackage{hyperref}
\usepackage{tikz}
\usepackage{mathtools}
\usepackage{multirow}

\newtoggle{Arxiv}

\toggletrue{Arxiv}

\newcommand{\version}[2]{%
\iftoggle{Arxiv}{%
#1
}
{#2}
}%

\acrodef{CSI}{channel state information}
\acrodef{CRN}{chemical reaction network}
\acrodef{GRN}{gene regulatory network}
\acrodef{BM}{Boltzmann machine}
\acrodef{MC}{molecular communication}
\acrodef{IoBNT}{Internet of Bio-Nano Things}
\acrodef{ML}{maximum-likelihood}
\acrodef{MAP}{maximum-a-posteriori}
\acrodef{ISI}{inter-symbol interference}
\acrodef{BER}{bit error rate}
\acrodef{CTMC}{continuous-time Markov Chain}
\acrodef{SNR}{signal-to-noise ratio}
\acrodef{RV}{random variable}
\acrodef{CSK}{concentration shift keying}
\acrodef{MoSK}{molecular shift keying}
\acrodef{BCSK}{binary concentration shift keying}
\acrodef{FVBM}{fully visible Boltzmann machine}
\acrodef{iid}[i.i.d.]{independent and identically distributed}
\acrodef{wert}[w.r.t.]{with regard to}
\acrodef{SCRN}{stochastic chemical reaction network}
\acrodef{iff}{if and only if}
\acrodef{EM}{electro-magnetic}
\acrodef{RX}{receiver}
\acrodef{TX}{transmitter}

\acrodef{AMV}{approximate majority voting}
\acrodef{MNRM}{modified next reaction method}

\acrodef{wrt}[w.r.t.]{with respect to}

\acrodef{PMES}[PM-ES]{pilot mode external signal}
\acrodef{DMES}[DM-ES]{data mode external signal}
\acrodef{STES}[ST-ES]{symbol trigger external signal}

\acrodef{DTMC}{discrete-time Markov chain}

\renewcommand{\vec}[1]{\ensuremath{\bm{\mathrm{#1}}}}
\newcommand{\Rvec}[1]{\ensuremath{\underline{\bm{\mathrm{#1}}}}}
\newcommand{\mol}[1]{\ensuremath{\mathrm{#1}}}

\newcommand{\xhat}{\hat{x}}
\newcommand{\xhatmap}{\hat{x}^{\mathrm{MAP}}}

\newcommand{\numap}{\nu^{\mathrm{MAP}}}

\newcommand{\Xhat}{\hat{X}}
\newcommand{\y}{\mathbf{y}}
\newcommand{\Y}{\mathbf{\underline{Y}}}
\newcommand{\Yl}{\underline{\bm{\mathrm{Y}}}[l]}
\newcommand{\Yli}{\left[\underline{\bm{\mathrm{Y}}}[l]\right]_i}
\newcommand{\z}{\mathbf{z}}
\newcommand{\Z}{\mathbf{\underline{Z}}}
\newcommand{\partition}{\mathcal{Z}}

\newcommand{\expectation}[2]{\mathbb{E}_{#1}\{#2\}}
\newcommand{\Var}[1]{\mathrm{Var}\{#1\}}

\newcommand{\XtrueOn}{\mathrm{X^{ON}_{true}}}
\newcommand{\XtrueOff}{\mathrm{X^{OFF}_{true}}}

\newcommand{\Xhaton}{\mathrm{\hat{X}^{ON}}}
\newcommand{\Xhatoff}{\mathrm{\hat{X}^{OFF}}}
\newcommand{\XhatCon}{\mathrm{\hat{X}_{ON}^C}}
\newcommand{\XhatCoff}{\mathrm{\hat{X}_{OFF}^C}}

\newcommand{\Yon}[1]{\mathrm{Y}^{\mathrm{ON}}_{#1}}
\newcommand{\Yoff}[1]{\mathrm{Y}^{\mathrm{OFF}}_{#1}}
\newcommand{\Zon}[1]{\mathrm{Z}^{\mathrm{ON}}_{#1}}
\newcommand{\Zoff}[1]{\mathrm{Z}^{\mathrm{OFF}}_{#1}}
\newcommand{\Won}{\mathrm{W}}

\newcommand{\BER}{\mathrm{BER}}

\newcommand{\MAPthreshold}{\nu^{\mathrm{MAP}}}

\newcommand{\nr}{n_{\mathrm{r}}}
\newcommand{\nrb}{n_{\mathrm{r,b}}}
\newcommand{\Nrb}{N_{\mathrm{r,b}}}
\newcommand{\YN}{\mathrm{Y}}
\newcommand{\nw}{n_{\mathrm{W}}}
\newcommand{\Nwa}{N_{\mathrm{W}}}
\newcommand{\nwa}{n_{\mathrm{W}}}

\newcommand{\wxy}{\mathrm{w}}

\newcommand{\cx}{c[l]}

\newcommand{\Deltac}{\Delta[l]}

\newcommand{\rrc}[1]{k_{\mathrm{#1}}} %

\newcommand{\xtrue}{ x_{ \mathrm{true} } }
\newcommand{\meter}{\mathrm{m}}
\newcommand{\second}{\mathrm{s}}

\newcommand{\crn}[1]{\mathcal{C}_{\mathrm{#1}}}
\newcommand{\kvec}{\mathbf{k}}

\newcommand{\transpose}{\mathrm{\scriptscriptstyle T}} %
\newcommand{\W}{\mathbf{W}}
\newcommand{\thetavec}{\boldsymbol{\theta}}
\newcommand{\V}{\mathbf{V}}

\newcommand{\kon}{k_{\mathrm{on}}}
\newcommand{\koff}{k_{\mathrm{off}}}

\newtheorem{theorem}{Theorem}

\newtheorem{definition}{Definition}

\newcommand{\propensity}{\rho}
\newcommand{\Amol}{\mathrm{A}}

\newcommand{\Bon}{\mathrm{B^{ON}}}
\newcommand{\Boff}{\mathrm{B^{OFF}}}

\newcommand{\Ta}{\tau_{\mathrm{a}}}

\newcommand{\Tsym}{\tau_{\mathrm{sym}}}

\newcommand{\Ton}{\mathrm{T^{ON}}}
\newcommand{\Toff}{\mathrm{T^{OFF}}}

\newcommand{\Son}{\mathrm{S^{ON}}}
\newcommand{\Sb}{\mathrm{S^{B}}}
\newcommand{\Soff}{\mathrm{S^{OFF}}}

\newcommand{\Ptrain}{\mathrm{P^{pilot}}}
\newcommand{\Ptransmit}{\mathrm{P^{data}}}

\newcommand{\Roff}{\mathrm{R^{OFF}}}
\newcommand{\Ron}{\mathrm{R^{ON}}}
\newcommand{\Rb}{\mathrm{R^{B}}}

\newcommand{\Hmol}{\mathrm{H}}

\newcommand{\Aoff}{\mathrm{A^{OFF}}}
\newcommand{\Aon}{\mathrm{A^{ON}}}

\newcommand{\Doff}{\mathrm{D^{OFF}}}
\newcommand{\Don}{\mathrm{D^{ON}}}
\newcommand{\Db}{\mathrm{D^{B}}}

\newcommand{\Loff}{\mathrm{L^{OFF}}}
\newcommand{\Lon}{\mathrm{L^{ON}}}
\newcommand{\Ioff}{\mathrm{I^{OFF}}}
\newcommand{\Ion}{\mathrm{I^{ON}}}

\newcommand{\transmatrix}{\mathbf{P}}
\newcommand{\pdistr}{\boldsymbol{\pi}}
\newcommand{\pdistrsteadystate}{\boldsymbol{\pi}^{\infty}}
\newcommand{\nwtotal}{n_{\mathrm{W}}^{\mathrm{total}}}

\newcommand{\uc}{\frac{\mathrm{molecules}}{\si{\cubic\metre}}}

\newcommand{\ponrel}{p_{\mathrm{rel,on}}}

\newcommand{\nsamples}{n_{\mathrm{t}}}

\usepackage{enumitem}
\newlist{consideration}{enumerate}{1}
\setlist[consideration]{label=\textbf{(C\arabic*)}, ref=\textbf{C\arabic*}}

\newlist{scenario}{enumerate}{1}
\setlist[scenario]{label=\textbf{(S\arabic*)}, ref=\textbf{S\arabic*}}

\allowdisplaybreaks

\makeatletter
\long\def\@makecaption#1#2{\ifx\@captype\@IEEEtablestring%
    \footnotesize\begin{center}{\normalfont\footnotesize #1}\\
        {\normalfont\footnotesize\scshape #2}\end{center}%
    \@IEEEtablecaptionsepspace
    \else
    \@IEEEfigurecaptionsepspace
    \setbox\@tempboxa\hbox{\normalfont\footnotesize {#1.}~~ #2}%
    \ifdim \wd\@tempboxa >\hsize%
    \setbox\@tempboxa\hbox{\normalfont\footnotesize {#1.}~~ }%
    \parbox[t]{\hsize}{\normalfont\footnotesize \noindent\unhbox\@tempboxa#2}%
    \else
    \hbox to\hsize{\normalfont\footnotesize\hfil\box\@tempboxa\hfil}\fi\fi}
\makeatother

\renewcommand{\baselinestretch}{1.2}

\begin{document}
\bstctlcite{IEEEexample:BSTcontrol}

\title{\vspace*{-10mm}Closing the Implementation Gap in MC:\\Fully Chemical Synchronization and Detection for Cellular Receivers
\vspace*{-0.5cm}
}
\author{Bastian~Heinlein,
        Lukas~Brand,
        Malcolm~Egan,
        Maximilian~Schäfer,\\
        Robert~Schober,
        and~Sebastian~Lotter%
\thanks{Bastian Heinlein, Lukas Brand, Maximilian Schäfer, Robert Schober, and Sebastian Lotter are with Friedrich-Alexander-Universität Erlangen-Nürnberg, Erlangen, Germany. Malcolm Egan is with Univ. Lyon, Inria, INSA Lyon, CITI, Villeurbanne, France.}
\thanks{This paper was presented in part at the ACM Nanoscale Computing and Communications Conference, 2023 \cite{bahe_nanocom}.}
}

\markboth{}%
{Heinlein \MakeLowercase{\textit{et al.}}: Designing a Cellular Receiver for Molecular Communications}
\maketitle
\vspace*{-1.6cm}
\begin{abstract}
\vspace*{-0.2cm}
In the context of the \ac{IoBNT}, nano-devices are envisioned to perform complex tasks collaboratively, i.e., by communicating with each other. One candidate for the implementation of such devices are engineered cells due to their inherent biocompatibility. However, because each engineered cell has only little computational capabilities, transmitter and \ac{RX} functionalities can afford only limited complexity.
In this paper, we propose a simple, yet modular, architecture for a cellular \ac{RX} that is capable of processing a stream of observed symbols using chemical reaction networks. Furthermore, we propose two specific detector implementations for the \ac{RX}. The first detector is based on a machine learning model that is trained {\em offline}, i.e., before the cellular \ac{RX} is deployed. The second detector utilizes pilot symbol-based training and is therefore able to continuously adapt to changing channel conditions {\em online}, i.e., after deployment.
To coordinate the different chemical processing steps involved in symbol detection, the proposed cellular \ac{RX} leverages an internal chemical timer.
Furthermore, the \ac{RX} is synchronized with the transmitter via external, i.e., extracellular, signals.
Finally, the proposed architecture is validated using theoretical analysis and stochastic simulations. The presented results confirm the feasibility of both proposed implementations and reveal that the proposed online learning-based \ac{RX} is able to perform reliable detection even in initially unknown or slowly changing channels.
By its modular design and exclusively chemical implementation, the proposed \ac{RX} contributes towards the realization of versatile and biocompatible nano-scale communication networks for \ac{IoBNT} applications narrowing the existing implementation gap in cellular \acf{MC}.

\vspace*{-0.4cm}
\end{abstract}
\acresetall

\section{Introduction}
\label{sec:introduction}
\vspace*{-3mm}
In contrast to conventional communication based on electromagnetic waves, \ac{MC} concerns the transmission of information using molecules.
As such, \ac{MC} is expected to facilitate a variety of applications that require communication among biological nanomachines such as lab-on-a-chip, environmental monitoring, and improved detection and treatment of diseases \cite{nakano_mc_opportunities_applications, akyildiz_panacea_infectious_diseases}.
In many of these applications, synthetic, i.e., engineered, cells acting as sensors and actuators, are envisioned to be integrated within the \ac{IoBNT} to collaboratively solve complex tasks \cite{haselmayr_iobnt}. 

Collaboration requires communication. To enable this communication, numerous \ac{RX} and \ac{TX} designs have been proposed in the \ac{MC} literature and several works studied the communication performance of optimal and sub-optimal \acp{RX} (e.g., \cite{noel_optimal_RX, li_clock_free,li_low_complexity_non_coherent_RX}). In some scenarios, the communication channel is not known before the deployment of the communication devices or time-variant, e.g., due to \ac{RX} or \ac{TX} mobility \cite{ahmadzadeh_mobile_mc} or changes in the level of background noise arising from the activities of other \acp{TX}. In these cases, non-coherent or adaptive detection schemes are required \cite{jia2019_concentration_difference,damrath2016_adaptive_threshold_Detection,alshammri_adaptive_neuro_fuzzy_receivers}. 
All the aforementioned existing \ac{RX} designs rely on the support of some, either external or not further specified, computation unit that is able to execute the mathematical functions required for symbol detection. However, to unleash the full potential of \ac{MC}, e.g., for \ac{IoBNT} applications, the required computations need to be performed locally, e.g., inside the synthetic cells acting as the \ac{MC} \acp{RX}. So far, there exists a large implementation gap for such cellular \ac{MC} systems that results from a severe mismatch between the computational requirements of the theoretical \ac{RX} designs proposed by communication engineers and the capabilities of synthetic cellular \acp{RX} that were realized in existing testbeds \cite{experimental_research}.

Several research efforts to narrow this gap have been reported in the literature.
Bio-inspired, adaptive \acp{RX} that utilize adaptation strategies as found in natural cells, such as the tuning of ligand-receptor response curves, were proposed in \cite{kuscu_adaptive_receiver_preprint}. On the other hand, a large body of (theoretical) literature investigated the possible realization of specific computational functions in nanomachines. 
The related literature can be divided into two main groups \cite{kuscu_proceedings_review}.
The first group is based on nanomaterials. In \cite{kuscu_biofets}, the use of nanoscale field effect transistor-based electrical biosensors (bioFETs) as \acp{RX} was proposed. However, their temperature sensitivity \cite{kuscu_sinw_fet} and inability to detect electrically neutral ligands \cite{aktas_mechanical_RX} make the usage of bioFETs in in-body applications challenging. Related approaches try to alleviate some of these issues, e.g., flexure field-effect transistors facilitate the detection of electrically neutral ligands \cite{aktas_mechanical_RX}. 
The second group is based on biological cells and can be further divided into \ac{GRN}-based and \ac{CRN}-based implementations. \acp{GRN} rely on the expression of genes, i.e., the production of proteins: Since produced proteins can in turn promote or suppress the expression of genes, a complex control system can be created and computational units like switches can be realized \cite{gardner_genetic_toggle_switch}. Because \acp{GRN} can be designed using tools from synthetic biology, they present a promising approach to perform computations in engineered cells, e.g., to implement the functionalities required for an \ac{RX}. In the \ac{MC} literature, several \ac{GRN}-based \acp{RX} were proposed for detection \cite{amerizadeh_bacterial_rx,unluturk_GRN_biotransceivers,durmaz_engineering_genetic_cicuits_rx,jadsadaphongphaiboo_bcsk_rx_tx_model} and error-correction coding \cite{yu2022convolutional,marcone_parity_check_GRNs}. Furthermore, a system-theoretic characterization of \ac{GRN} models was provided in \cite{pierobon_molecular_system_model_grns,pierobon_systems_theoretical_model_grns}. However, even though it is in principle possible to equip biological cells with pre-designed \acp{GRN}, this approach has limitations. First, \acp{GRN} operate rather slowly and even the execution of simple computations can take several hours \cite{gardner_genetic_toggle_switch}. Second, the \ac{GRN} designs proposed in the \ac{MC} literature are complex as they typically rely on the combination of several digital functions. Thus, their practical realization is in general very challenging \cite{purcell_synthetic_grns,wang_riboswitches}.

Similar to \acp{GRN}, \acp{CRN} usually rely on proteins to perform computations when used in cells. However, \acp{CRN} employ already produced proteins and, hence, do not require the synthesis of new proteins\footnote{For this reason we consider \acp{CRN} and \acp{GRN} as different concepts in this work, even though \acp{GRN} can be technically interpreted as a special case of \acp{CRN}, where the chemical reactions involved in protein synthesis are most prevalent.}. In nature, \acp{CRN} perform computations, for example, through the phosphorylation and de-phosphorylation of certain proteins, e.g., in bacterial chemotaxis \cite{wadhwa_nature_reviews_microbio_chemotaxis}. Recent progress in the engineering of such \textit{post-translational} circuits \cite{chen_programmable_protein_circuits} renders \acp{CRN} a viable option for the implementation of biological \acp{RX} that can realize higher computation speeds compared to \ac{GRN}-based \acp{RX}.
However, the practical feasibility of \ac{CRN}-based cellular \acp{RX} hinges on the following practical constraints:
\begin{consideration}
    \item \Ac{RX} design in \ac{MC} can in general not rely on the existence of mathematically tractable \textbf{channel models}, since these are notoriously difficult to obtain for many \ac{MC} channels, e.g., due to complex environment geometries or unknown physical parameters. In some situations, simulation or experimental data may be available and utilized to train a coherent detector offline, in other scenarios, when the channel is unknown \textit{a priori} or time-variant, only non-coherent detection schemes are suitable. In both cases, the absence of tractable channel models poses a particular challenge for the design of cellular \ac{CRN}-based \acp{RX}, since neither data-driven nor (pilot) training-based methods have yet been explored in the context of chemical \acp{RX}.\label{con:adaptive_RX}

    \item \ac{TX} and \ac{RX} need to be \textbf{synchronized} to ensure that symbol transmission and detection happen at the right time. While this requirement exists similarly for non-cellular \ac{MC} systems, it has not been addressed yet for cellular \acp{RX}, i.e., in the absence of external computation units.\label{con:synchronization}
        
    \item To facilitate its chemical implementation, any practical \ac{CRN}-based \ac{RX} must be constrained to a \textbf{simple architecture} with only few chemical species and reactions involved. A large number of chemical species and/or reactions would prohibitively increase the implementation complexity.\label{con:practical_implementation}
\end{consideration}

Some works in the \ac{MC} literature already considered \acp{CRN} to realize single \ac{RX} functionalities. On the protocol level, the use of enzyme-based computing was studied in \cite{walsh_synthetic_protocols_enzyme_dna_based_computing}. In \cite{cao_crn_based_detection,hamdan_generalized_solution,kuscu_ml_detectors}, \acp{CRN} were proposed to process the received signal for further analysis. A \ac{CRN} making an actual decision about the received symbol was designed in \cite{egan_biological_circuits_mosk}. Further \ac{CRN}-based detector implementations were reported in \cite{chou_molecular_circuits_approx_map,riaz_map_spatially_partioned,chou2012_molecular_circuits_frequency_coded}. The use of chemical reactions in micro-fluidic \ac{MC} systems to implement \ac{RX} and \ac{TX} functionalities was proposed in \cite{bi_crn_microfluidic_tx_rx} and experimentally validated in \cite{walter_ncomms}.
However, none of the existing works in the \ac{MC} literature fulfills all the constraints \ref{con:adaptive_RX}-\ref{con:practical_implementation} at the same time.

In the conference version of this paper \cite{bahe_nanocom}, two \ac{CRN}-based detectors were proposed that fulfill constraints \ref{con:adaptive_RX} and \ref{con:practical_implementation} mentioned above.
Here, we extend the preliminary results from \cite{bahe_nanocom} towards a fully chemical \ac{RX} implementation meeting all constraints. The main contributions of this paper can be summarized as follows.
\begin{itemize}
    \item The \ac{CRN}-based detectors proposed in \cite{bahe_nanocom} are extended by a chemical thresholding mechanism that aggregates the randomly fluctuating \ac{CRN} output into a detection decision. 
    \item The adaptive detector design from \cite{bahe_nanocom} is complemented in this paper by a theoretical steady state error analysis. 
    \item In contrast to \cite{bahe_nanocom}, the \ac{RX} design proposed in this paper includes fully chemical synchronization and timing mechanisms that ensure the execution at the right time and the correct order, respectively, of all required computations at the \ac{RX}. These mechanisms in particular enable the \ac{RX} proposed in this paper to reset its chemical state after each symbol detection and seamlessly process a stream of received symbols, which was not considered in \cite{bahe_nanocom}.
    \item The adaptive \ac{RX} proposed in \cite{bahe_nanocom} is extended so it can be toggled between two operating modes, the learning and the data transmission mode.
    \item Finally, numerical results from extensive stochastic simulations confirm that the proposed \ac{CRN}-based \ac{RX} design delivers competitive performance in terms of the achieved \ac{BER} when compared to the \ac{MAP} detector.
\end{itemize}
In summary, in this work, we propose a novel \ac{CRN}-based cellular \ac{RX} architecture that may contribute towards the practical realization of cellular \acp{RX}, hence narrowing the existing implementation gap for cellular \ac{MC} systems.

The remainder of this paper is organized as follows: The system model and important preliminaries on \acp{CRN} are summarized in Section~\ref{sec:system_overview}. The proposed machine learning model-based and adaptive detectors are introduced in Sections~\ref{sec:BMbasedRX} and \ref{sec:adaptiveRX}, respectively. In Section~\ref{sec:timing}, the chemical implementation of the timing mechanism is discussed. The proposed detection algorithm and the end-to-end \ac{CRN}-based \ac{RX} implementation are validated with extensive simulations in Section~\ref{sec:simulations}. Finally, Section~\ref{sec:conclusion} summarizes the main findings and future research directions.

\textit{Notation:} Vectors and matrices of constants are denoted by lowercase and uppercase bold letters, respectively. The $i$-th entry of vector $\vec{x}$ is denoted as $\left[\vec{x}\right]_i$. The transpose of matrix $\mathbf{A}$ is denoted by $\mathbf{A}^{\transpose}$. Scalar \acp{RV} and vectors of \acp{RV} are denoted by uppercase and bold underlined uppercase letters, respectively. The expectation operator for a probability distribution $q_{\Z}(\z)$ is denoted by $\expectation{q_{\Z}}{\cdot}$. The variance of \ac{RV} $X$ is represented by $\Var{X}$. \mbox{Sets are denoted by calligraphic letters and} the cardinality of a set $\mathcal{A}$ is denoted by $|\mathcal{A}|$.  Chemical species are represented \mbox{by roman uppercase letters}.
\vspace*{-0.7cm}
\section{System Model and Preliminaries}
\vspace*{-0.2cm}
\label{sec:system_overview}
\begin{figure}
    \begin{minipage}{0.48\textwidth}
        \centering
        \vspace*{-10mm}
        \includegraphics[width=3.2in]{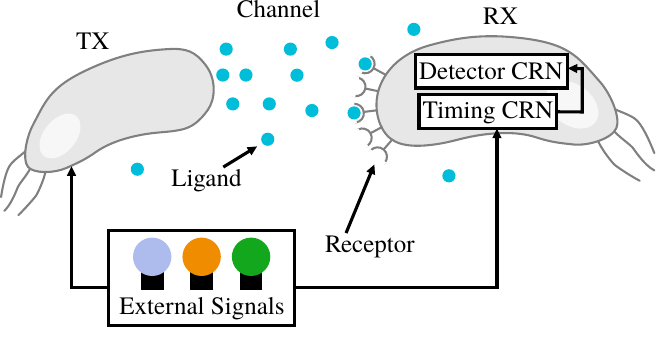}
        \vspace*{-3mm}
        \caption{\textbf{System Model.} External signals synchronize the transmission and reception of data and, optionally, pilot symbols. The \ac{RX} has an internal timing mechanism to coordinate the detection and training process prompted by the respective external signals.}
        \label{fig:system_model}
    \end{minipage}
    \hfill
    \begin{minipage}{0.48\textwidth}
        \centering
        \vspace*{-5mm}
        \includegraphics[width=3.2in]{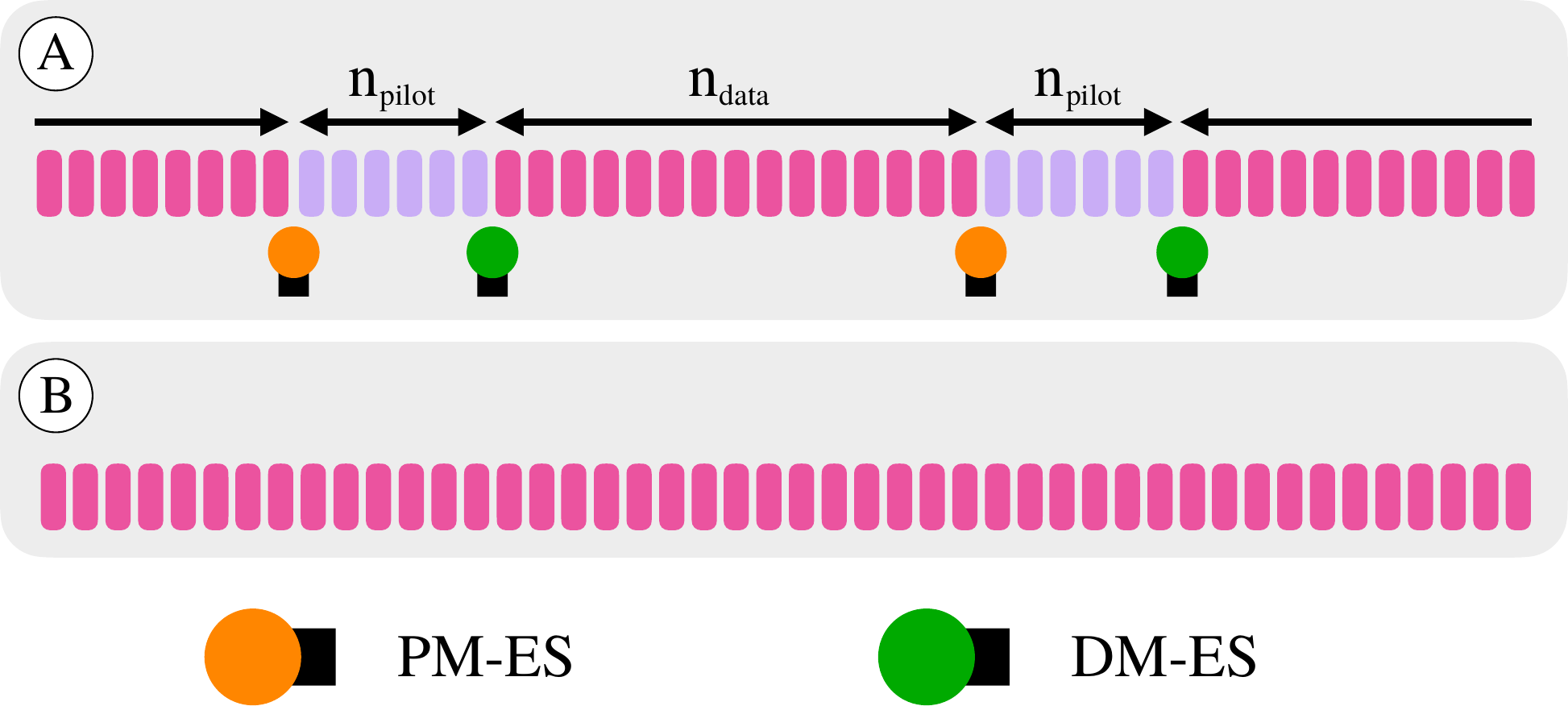}
        \vspace*{-4mm}
        \caption{\textbf{Data and Pilot Mode.} A: \Acf*{PMES} and \acf*{DMES} switch the \ac{TX} (\ac{RX}) into the \textit{pilot mode} and the \textit{data mode} for the transmission (reception) of pilot and data symbols, respectively. B: If a detector does not use pilot symbols, only data symbols are transmitted.}
    \label{fig:transmission_phases}
    \end{minipage}
    \vspace*{-1cm}
\end{figure}
In this section, we review some preliminaries on \acp{CRN} that are later used in Sections~\ref{sec:BMbasedRX}-\ref{sec:timing}, describe the system model, and provide an overview of the proposed \ac{RX} architecture.
\vspace*{-0.5cm}

\subsection{Chemical Reaction Network Basics}
\label{sec:CRNprimer}
\vspace*{-2mm}
Formally, a \ac{CRN} $\crn{}$ is defined by a tuple $\crn{} = (\mathcal{S}, \mathcal{R}, \kvec)$, where $\mathcal{S}$, $\mathcal{R}$, and $\kvec$ are the set of chemical species, the set of reactions defined over $\mathcal{S}$, and the vector of reaction rate constants, respectively. At each time $t$, the state of the \ac{CRN} is defined by the number of molecules of each species $\mol{S} \in \mathcal{S}$, denoted by $n_{\mol{S}}(t)$\footnote{The dependency of $n_{\mol{S}}(t)$ on $t$ is suppressed in the notation in the following, i.e., we simply write $n_{\mol{S}}$ instead, when it is not relevant in the particular context.}. Driven by chemical reactions, the \ac{CRN} transitions between different states and can in this way perform computations. 

Consider a minimal \ac{CRN} with $\mathcal{S}=\{\Amol, \Bon, \Boff\}$ and one chemical reaction with rate constant $\rrc{1}>0$, i.e., $\kvec=[\rrc{1}]$. A reaction and its corresponding rate constant are typically summarized by a chemical equation, as, for example, the following:
\begin{equation}
    \Amol + \Bon \xrightarrow{\rrc{1}} \Amol + \Boff.\label{eq:crn_example}
\end{equation}
In \eqref{eq:crn_example}, one molecule of type $\Amol$ reacts with one molecule of type $\Bon$ to form one molecule $\Amol$ and one molecule $\Boff$ and the reaction occurs with rate constant $\rrc{1}$.
Here, $\Amol$ acts as \textbf{catalyst}, i.e., it is required for the conversion of the \textbf{substrate} $\Bon$ to the \textbf{product} $\Boff$, but is not consumed by the reaction.
Catalytic reactions are omnipresent in nature and, for example, based on enzymes as catalyst. Also in the \ac{RX} architecture proposed in this paper, such catalytic reactions are an integral part\footnote{We also utilize reactions with more than two reacting molecules (reactants). However, such reactions can be decomposed into several simpler reactions with only two reactants.}.
Furthermore, the notations $\Bon$ and $\Boff$ are indicative of the activated and inactivated states of a molecule $\mathrm{B}$, respectively.
In nature, these states could correspond to, for example, the phosphorylated and dephosphorylated states of $\mathrm{B}$.

We assume mass action kinetics in this paper, i.e., the probability of any chemical reaction to happen is assumed to be proportional to the number of the reacting molecules \cite{mass_action_kinetics} within the reaction volume in which the \ac{CRN} operates. Then, for example, the probability $\propensity_1(t)$ of the reaction in \eqref{eq:crn_example} to occur per unit time is given by
\begin{equation}
    \propensity_1(t) = \rrc{1} \cdot n_{\Amol}(t) \cdot n_{\Bon}(t),\label{eq:mass_action_example}
\end{equation}
and we refer to \cite[Eq. (6)]{poole2017_chemical_BMs} for a more general form of \eqref{eq:mass_action_example}. 
When engineering a \ac{CRN} for a particular functionality (as it is done in this paper), a particular value of $\rho_1(t)$ may be required.
In practice, this can be achieved by tuning $\rrc{1}$ or, alternatively, by adjusting $n_{\Amol}(t)$. Also, in a \ac{CRN} with multiple reactions, the rate constants can be multiplied by any positive constant without altering the system dynamics except for the time scale. This is possible because the equations governing the state evolution of a \ac{CRN} are linear in the reaction rate constants \cite{mass_action_kinetics}. Hence, the rate constants are unique only up to a common scaling factor and we employ normalized reaction rate constants and normalized time scales in the remainder of this paper.

Now, chemical reactions can in general not only be catalyzed by enzymes as, for example, in \eqref{eq:crn_example}, but also by \textbf{external signals}, i.e., optical or chemical signals originating from outside the reaction volume.
In its simplest form, while an external signal is ON, a chemical reaction is enabled, i.e., the reaction occurs with some non-zero rate constant; when the signal is OFF, the respective reaction does not occur.
Chemical reactions involving external signals have been already described and even practically implemented in the \ac{MC} literature. One commonly used external signal in \ac{MC} testbeds is light \cite{experimental_research}, because it can easily be turned ON and OFF.
Furthermore, light sources with different wavelengths can be used to catalyze different chemical reactions. For example, \textit{green fluorescent protein variant ''Dreiklang''} is a protein that can be switched between different states by applying light of different wavelengths \cite{media_modulation}.
Light is also often used in synthetic biology, e.g., to trigger the release of cargo carried by liposomes \cite{chude_bio_cyber_interface} or to influence intra-cellular processes like gene-expression \cite{hartmann_control_gene_expression_with_light}.
Chemical external signals are, for example, well-known in the context of neuromodulation, where intracellular chemical processes in neural cells are adjusted depending on the presence of chemical substances in the extracellular matrix.

In the remainder of this paper, external signals are used for synchronizing \ac{TX} and \ac{RX}.
Formally, the impact of an external signal on a specific reaction is reflected in a \mbox{time-dependent} reaction rate constant and, since we exclusively consider external signals that are either ON or OFF, these time-dependent reaction rate constants are piecewise constant.
In terms of notation, if an external signal $\mathrm{e}$, in short ES-$\mathrm{e}$, e.g., light of a specific wavelength, enables the conversion of a molecule of type $\mathrm{B}$ from state $\Boff$ to $\Bon$ with rate constant $\rrc{e}$, we denote this in the following as
\begin{equation}\label{eq:system_model:reaction_external_catalyst}
    \Boff  \xrightarrow{\textrm{ES-}\mathrm{e}\,\, \rrc{e}} \Bon.
\end{equation}
\vspace*{-1.4cm}

\subsection{Transmitter and Channel Models}
\label{sec:system_model}
We consider \ac{MC} between a single cellular \ac{TX} and a single cellular \ac{RX} as depicted in Fig.~\ref{fig:system_model}.
To convey information, the \ac{TX} cell releases signaling molecules into the fluid environment which propagate towards the \ac{RX} cell where they are sensed by receptor proteins embedded in the \ac{RX} cell's membrane\footnote{In this paper, the proposed detector designs do not rely on specific assumptions on the physical phenomena underlying the signaling molecule propagation from the \ac{TX} cell towards the \ac{RX} cell.}.
Specifically, for time-slotted transmission in which index $l \in \mathbb{N}_0$ denotes the symbol interval, the \ac{TX} utilizes binary concentration shift keying to transmit binary symbols $X[l] \in \lbrace 0, 1 \rbrace$. In response, the \ac{RX} cell is exposed to signaling molecule concentration $C[l]$,
\vspace*{-0.1cm}
\begin{equation}
    C[l] = \Deltac X[l] + N[l],\label{eq:system_model_concentrations}
\end{equation}
where $\Deltac>0$ denotes the concentration of signaling molecules at the \ac{RX} due to the release of molecules by the \ac{TX} and $N[l] \geq 0$ denotes the concentration of background noise molecules of the same type as the signaling molecules at the \ac{RX} in symbol interval $l$.
In \eqref{eq:system_model_concentrations}, we assume that consecutive symbol intervals do not affect each other, i.e., $C[l]$ depends only on $X[l]$ (and $l$) and is independent of any $X[l']$, $l' \neq l$\footnote{This is the case, e.g., for sufficiently long symbol intervals. The required length of the symbol interval can be shortened, e.g., by releasing enzymes into the channel to degrade the signaling molecules \cite{noel_enzymes}.}.

The cellular \ac{RX} estimates the ligand concentration in its environment using $\nr$ cell surface receptors (cf. Figure~\ref{fig:system_model}), where, in line with practical considerations, each receptor is assumed to remain in one single state throughout the detection interval (i.e., either bound or unbound) \cite{bernetti_binding_kinetics_review}.
Consequently, the receptor occupancy in symbol interval $l$ is denoted by random vector $\Y[l] \in \{0,1\}^{\nr}$, where $\Yli = 1$ \ac{iff} the $i$-th receptor is bound in symbol interval $l$.
Thus, the considered \ac{MC} channel in interval $l$ is fully characterized by the joint distribution $q_{\Y[l],X[l]}(\y, x) = \Pr\left[ \Y[l]=\y, X[l]=x \right]$ of the receptor states $\Y[l]$ and transmitted symbol $X[l]$.

For any practical \ac{MC} system, the channel parameters $\Deltac$ and $N[l]$ as well as $q_{\Y[l],X[l]}$ are difficult if not impossible to obtain and in fact the \ac{RX} proposed in this paper does not rely on knowing either one of them.
However, in the special case that all $\nr$ receptors are chemically identical and the number of ligands is large as compared to the number of receptors, since we assumed that all receptors are exposed to the same ligand concentration $C[l]$, the receptor occupancies $\Yli$ are \ac{iid}. Hence, the probability of a receptor to be occupied can be derived as \cite{unluturk_GRN_biotransceivers}
\vspace*{-0.1cm}
\begin{equation}
    \Pr\left[\Yli=1|X[l]\right] = \frac{C[l]}{C[l]+\frac{\rrc{-}}{\rrc{+}}},\quad \forall i=1,\ldots,\nr,\label{eq:receptor_binding_probability}
\end{equation}
where $\rrc{+}$ and $\rrc{-}$ denote the receptor binding and unbinding rate constants, respectively.

\noindent Although the \ac{RX} design proposed in this paper is not limited to the cases when \eqref{eq:receptor_binding_probability} applies, we use \eqref{eq:receptor_binding_probability} to generate the receptor occupancy realizations in the computer simulations presented Section~\ref{sec:simulations}.

In the following sections, we often suppress the dependency of $\Yl$ and $X[l]$ on $l$ in the notation and write $\Rvec{Y}$ and $X$ instead, where it is understood that the respective considerations for $\Rvec{Y}$ and $X$ hold for any $\Yl$ and $X[l]$.

\subsection{CRN-based MAP Detection}\label{sec:system_model:crn-based_detection}

\subsubsection{Direct Realization}

The computational substrate of all computing units of the \ac{RX} proposed in this paper are molecule species that chemically interact with each other; consequently, the inputs and outputs of each \ac{CRN}-based computation are represented by corresponding molecule types and counts.
Now, as we have discussed in Section~\ref{sec:CRNprimer}, chemical reactions are probabilistic events, in particular if molecule counts are low, such as in \acp{CRN} operating within single cells.
Let $\mathcal{A} = \left\lbrace \mathrm{A}_1,\ldots,\mathrm{A}_{|\mathcal{A}|} \right\rbrace \subset \mathcal{S}$ and $\mathcal{B} = \left\lbrace \mathrm{B}_1,\ldots,\mathrm{B}_{|\mathcal{B}|} \right\rbrace \subset \mathcal{S}$ denote the input and output species of some computation executed by \ac{CRN} $\mathcal{C}=\left(\mathcal{S}, \mathcal{R}, \bm{\mathrm{k}}\right)$.
Then, the time-dependent vectors of species counts $\Rvec{N}_{\mathcal{A}}(t) = [N_{\mathrm{A}_1}(t),\ldots,N_{\mathrm{A}_{|\mathcal{A}|}}(t)]$ and $\Rvec{N}_{\mathcal{B}}(t) = [N_{\mathrm{B}_1}(t),\ldots,N_{\mathrm{B}_{|\mathcal{B}|}}(t)]$ are discrete-valued random processes and the computation executed by $\mathcal{C}$ is probabilistic, i.e., expressed as a conditional probability distribution that we denote as $q_{\Rvec{N}_{\mathcal{B}}(t)|\Rvec{N}_{\mathcal{A}}(0)}(\vec{n}_{\mathcal{B}}(t)|\vec{n}_{\mathcal{B}}(0),\vec{n}_{\mathcal{A}}(0))$ \cite{virinchi_cell_infer_environment}\footnote{We will always assume deterministic initial conditions for the species in $\mathcal{B}$, i.e., $\bm{n}_{\mathcal{B}}(0)$ is fixed and, hence, suppressed in the notation in the following.}.
The idea behind \ac{CRN}-based \ac{MAP} detection is to design $\mathcal{C}$ in such a way that $q_{\Rvec{N}_{\mathcal{B}}(t)|\Rvec{N}_{\mathcal{A}}(0)}$ generates (approximately) \ac{iid} samples $\Rvec{N}_{\mathcal{B}}(t_s)$, $s = 1,\ldots,S$, according to the posterior distribution $\Pr\left[ X | \Y \right]$, where $S$ denotes the number of samples and we recall from the previous section that $\Y$ remains fixed during the entire detection process of a single symbol.
Here, $\mathcal{C}$ is said to generate samples from  $\Pr\left[ X | \Y \right]$ if there exist deterministic mappings $f$, $g$ that map each realization of random vector $\Y$, $\vec{y}$, to some input species count vector $\bm{n}_{\mathcal{A}}$, i.e., $f:\lbrace 0, 1 \rbrace^{\nr} \mapsto \mathbb{N}_0^{|\mathcal{A}|}$, and each possible realization of \ac{RV} $X$, $x$, to some output species count vector $\bm{n}_{\mathcal{B}}$, i.e., $g: \left\lbrace 0, 1 \right\rbrace \mapsto \mathbb{N}_0^{|\mathcal{B}|}$, respectively, and, for each $t_s$,
\begin{equation}
    q_{\underline{\bm{\mathrm{N}}}_{\mathcal{B}}(t_s)|\underline{\bm{\mathrm{N}}}_{\mathcal{A}}(0)}(\bm{n}_{\mathcal{B}}(t_s)| \bm{n}_{\mathcal{A}}(0) ) = q_{\underline{\bm{\mathrm{N}}}_{\mathcal{B}}(t_s)|\underline{\bm{\mathrm{N}}}_{\mathcal{A}}(0)}(g(x)|f(\vec{y})) = \Pr\left[ X=x | \Y=\vec{y} \right].
\end{equation}
Furthermore, $g$ needs to be bijective, i.e., $g^{-1}(\bm{n}_{\mathcal{B}}(t_s)) \in \lbrace 0, 1 \rbrace$ for each $\vec{n}_{\mathcal{B}}(t_s)$ with $\Pr\left[\vec{n}_{\mathcal{B}}(t_s)\right] > 0$, where $g^{-1}$ denotes the inverse mapping to $g$.
Then, the samples generated by $\mathcal{C}$, i.e., $\Rvec{N}_{\mathcal{B}}(t_1),\ldots,\Rvec{N}_{\mathcal{B}}(t_S)$, are statistically equivalent to \ac{iid} samples $X_1=g^{-1}(\Rvec{N}_{\mathcal{B}}(t_1)),\dots,X_S=g^{-1}(\Rvec{N}_{\mathcal{B}}(t_S))$ from the posterior and the \ac{MAP} estimate of $x$ given $\vec{y}$ is obtained from \ac{CRN} $\mathcal{C}$ as
\begin{equation}\label{eq:system_model:sample_averaging}
    \xhatmap = \begin{cases}
        1 & \text{if } \lim\limits_{S \to \infty} \frac{1}{S}\sum_{s=1}^S g^{-1}\left(\vec{n}_{\mathcal{B}}(t_s)\right) = \lim\limits_{S \to \infty} \frac{1}{S}\sum_{s=1}^S X_s \geq \frac{1}{2}, \\
        0 & \text{otherwise},
    \end{cases}
\end{equation}
due to the law of large numbers.
In short, for \ac{CRN}-based \ac{MAP} detection, $\mathcal{C}$ is designed to generate \ac{iid} samples from $\Pr\left[X|\Rvec{Y}\right]$, such that in the limit of a large number of samples the empirical mean of the samples converges to $\Pr\left[X=1|\Rvec{Y}\right]$.

\subsubsection{Equivalent Realization}

Rather than implementing the posterior $\Pr\left[X|\Rvec{Y}\right]$ directly as discussed in the previous section, the two \ac{CRN}-based detectors that are introduced in Sections~\ref{sec:BMbasedRX} and \ref{sec:adaptiveRX} make use of two observations. First, we observe from \eqref{eq:system_model:sample_averaging} that it is not necessary to perfectly know $\Pr[X|\Rvec{Y}]$ to make \ac{MAP} decisions. Instead, any probability mass function $p_{\Xhat|\Rvec{Y}}(x\,|\,\vec{y})$ defining a surrogate \ac{RV} $\Xhat$ can be employed for \ac{MAP} detection if it fulfills the condition of the following definition. 
\begin{definition}[MAP Property]\label{def:system_model:map_property}
Systems that generate samples of any \ac{RV} $\Xhat$ for which 
\begin{equation}
    p_{\Xhat|\Rvec{Y}}(x\,|\,\vec{y}) \geq \frac{1}{2} \textrm{ iff } \Pr[X=x\,|\,\Rvec{Y}=\vec{y}] \geq \frac{1}{2}, \quad \forall\, x \in \lbrace 0, 1 \rbrace, \y \in  \lbrace 0, 1 \rbrace^{\nr},
\end{equation}
possess the {\em MAP property}.
\vspace*{-0.3cm}
\end{definition}

Secondly, for \acp{CRN} that possess the \ac{MAP} property and have only two output species, i.e., $|\mathcal{B}|=2$, for which exactly one molecule of either one or the other species exists for all $t>0$, i.e., $n_{\mol{B}_1}(t) + n_{\mol{B}_2}(t) = 1$, $x=1$ ($x=0$) can be associated with the presence (absence) of a single molecule type, say $\mol{B}_1$.
In that case, \eqref{eq:system_model:sample_averaging} reduces to counting observations of $\mol{B}_1$, i.e.,
\begin{equation}\label{eq:system_model:sampling}%
    \xhat^{\mathrm{MAP}} = \begin{cases}
        1 & \text{if } \lim\limits_{S \to \infty} \frac{1}{S}\sum_{s=1}^S n_{\mol{B}_1}(t_s) \geq \frac{1}{2}, \\
        0 & \text{otherwise}.
    \end{cases}
\end{equation}
Motivated by this second observation and since it will prove useful in Sections~\ref{sec:BMbasedRX} and \ref{sec:adaptiveRX}, we identify all binary input and output variables in the detection process, i.e., $[\Y]_1,\ldots,[\Y]_{\nr}$ and $\hat{X}$, with the presence or absence of single molecular species in the respective detector \acp{CRN}.
To this end, we specialize the mappings $f$ and $g$ introduced before as
\begin{align}
    f(\vec{y}) &= [f_1([\y]_1),\ldots,f_{\nr}([\y]_{\nr})],
\end{align}
where
\vspace*{-0.3cm}
\begin{align}
    f_i(y) = \begin{cases}
        [n_{Y_i^{\mathrm{ON}}} = 1, n_{Y_i^{\mathrm{OFF}}} = 0] \textrm{ if } y = 1,\\
        [n_{Y_i^{\mathrm{ON}}} = 0, n_{Y_i^{\mathrm{OFF}}} = 1] \textrm{ otherwise},
    \end{cases}
\end{align}
and 
\vspace*{-0.3cm}
\begin{align}
    g(x) = \begin{cases}
        [n_{\Xhat^{\mathrm{ON}}} = 1, n_{\Xhat^{\mathrm{OFF}}} = 0] \textrm{ if } x = 1,\\
        [n_{\Xhat^{\mathrm{ON}}} = 0, n_{\Xhat^{\mathrm{OFF}}} = 1] \textrm{ otherwise},
    \end{cases}
\end{align}
and require for all $t$, $n_{Y_i^{\mathrm{ON}}}(t) + n_{Y_i^{\mathrm{OFF}}}(t) = 1$ for all $i = 1,\ldots,\nr$ and $n_{\Xhat^{\mathrm{ON}}}(t) + n_{\Xhat^{\mathrm{OFF}}}(t) = 1$.
Hence, for \ac{CRN}-based detectors, $\bigcup_{i=1,\ldots,\nr} \left\lbrace Y_i^{\mathrm{ON}}, Y_i^{\mathrm{OFF}} \right\rbrace$ and $\left\lbrace \Xhat^{\mathrm{ON}}, \Xhat^{\mathrm{OFF}}\right\rbrace$ constitute the sets of input and output species, respectively.
The chemical rationale behind the mapping $f$ is that depending on the state of each receptor (bound or unbound) its intracellular domain may exhibit different chemical characteristics that then correspond to the ON or OFF state of the associated molecule type, respectively.

Finally, we conclude this section noting that under the mild assumption that the \ac{RX} is agnostic \ac{wrt} the individual identities of the different receptors $[\Y]_i$, i.e., the receptors are statistically identical (not necessarily independent), the aggregate \ac{RV} $\Nrb = \sum_{i=1}^{\nr} Y_i$ can be used for detection instead of $\Y$.
In that case, if the likelihood $\Pr\left[\Nrb|X\right]$ is a unimodal distribution for all $x$, \ac{MAP}\footnote{Because we assume equiprobable information symbols, \ac{MAP} and \ac{ML} detection are equivalent.} detection can be performed by a threshold detector.
Specifically, the \ac{MAP} estimate $\xhatmap$ can be obtained from the receptor occupancy realization $\y$ in a single symbol interval using the optimal decision threshold $\numap$ as
\begin{equation}
    \xhatmap = \begin{cases}
        1 & \text{ if } \sum_{i=1}^{\nr} [\y]_i = \nrb \geq \numap, \\
        0 & \text{otherwise}.
    \end{cases}\label{eq:xhatmap_threshold}
\end{equation}
In the remainder of this paper, we assume that the \ac{MAP} detector reduces to a threshold detector and use the expressions $\Pr\left[\cdot|\Y=\y\right]$ and $\Pr\left[\cdot|\Nrb=\nrb\right]$ equivalently.
\vspace*{-0.5cm}

\subsection{Receiver Architecture}
The proposed \ac{RX} architecture can be divided into two components as illustrated in Figure~\ref{fig:system_model}:
\begin{enumerate}
    \item The task of the \textbf{detector unit} is to estimate $X[l]$ based on the receptor occupancy $\Yl$ in each symbol interval $l$. It comprises multiple chemical computation steps as detailed in the following sections. For the adaptive detector unit presented in Section~\ref{sec:adaptiveRX} the detection mechanism is complemented by a training mechanism that adapts the detector unit's decision threshold to time-varying channel conditions.
    \item The \textbf{internal timing unit} coordinates the temporal order of the different computation steps during the detection process and ensures that the \ac{RX} is reset once the detection of a symbol is finalized. For the proposed adaptive detector unit, the internal timing unit also coordinates the execution order of the computations during the training process.
\end{enumerate}

In the following sections, we first develop two different chemical detector units (Sections~\ref{sec:BMbasedRX} and \ref{sec:adaptiveRX}), before we present a fully chemical realization of the internal timing unit in Section~\ref{sec:timing}.
Utilizing these chemical implementations, the \ac{RX} proposed in this paper is able to operate fully locally, i.e., inside the \ac{RX} cell, and the output of the detector unit can directly be used by other chemical computational units in the \ac{RX} cell to process the received information further.

\vspace*{-0.5cm}
\subsection{Synchronization}

In order to synchronize the symbol transmission and reception between the \ac{TX} and the \ac{RX}, we exploit the fact that intracellular chemical reactions can be activated by external signals as discussed in Section~\ref{sec:CRNprimer}.
In our proposed system architecture, we utilize three different external signals.
The first one, the {\em \ac{STES}} (illustrated as a violet lamp in Fig.~\ref{fig:system_model} and throughout the remainder of this paper) initiates the symbol transmission, i.e., the molecule release, at the \ac{TX}, and triggers the detector unit at the \ac{RX}.
\ac{STES} is utilized by both \ac{CRN}-based detector units proposed in this paper.
The other two external signals, the {\em \ac{PMES}} and the {\em \ac{DMES}} (illustrated as orange and green lamps, respectively, in Fig.~\ref{fig:system_model} and throughout the remainder of this paper), are exclusively utilized by the proposed adaptive detector.
Specifically, the proposed adaptive detector relies on the repeated transmission of pilot symbol sequences.
Hence, \ac{TX} and \ac{RX} have to switch between a \textit{data mode} and a \textit{pilot mode} (they always remain in data mode for the non-adaptive detector).
\ac{PMES} initiates the switch from \textit{data mode} to \textit{pilot mode}, whereas \ac{DMES} initiates the switch from \textit{pilot mode} to \textit{data mode}, cf. Fig.~\mbox{\ref{fig:transmission_phases}-A}. In \textit{data mode}, the \ac{TX} transmits equiprobable data symbols, whereas it sends a fixed sequence of pilot symbols in \textit{pilot mode}.

\section{Boltzmann Machine-based Detector Implementation}
\label{sec:BMbasedRX}
In this section, we introduce a \ac{CRN}-based detection scheme derived from the chemical implementation of a \ac{BM}-based detector. The \ac{BM} is a model used in machine learning.
To this end, we first introduce some background on \acp{BM}, then discuss how \acp{BM} can be utilized for detection in the context of the considered \ac{MC} system model, and finally present the proposed \ac{CRN}-based implementation.
\vspace*{-0.5cm}

\subsection{Background on Boltzmann Machines}\label{sec:bm:background}
\Acp{BM} are graphical probabilistic models, i.e., any \ac{BM} is defined by a random vector of {\em binary-valued nodes} $\Z = [Z_1, \ldots, Z_m]$, where $Z_i \in \{0,1\}$ for all $i = 1, \ldots, m$\footnote{The variable $i$ is intentionally reused here, as the \ac{BM} nodes are later related to the receptors.}, the corresponding vector of {\em activation biases} $\thetavec \in \mathbb{R}^{m}$, and a symmetric, zero-diagonal {\em edge weight matrix} $\W \in \mathbb{R}^{m \times m}$.
The power of \acp{BM} lies in their ability to generate samples from any arbitrary joint probability distribution $q_{\Z'}(\z')$ of binary random variables $\Z' = [Z_1', \ldots, Z_n']$, where $Z_j' \in \{0,1\}$ for all $j = 1, \ldots, n$ and $n \leq m$, in a process called {\em Gibb's sampling}.
Gibb's sampling refers to repeatedly selecting a random node $i'$ from $\Z$ and updating its value according to the following rule
\begin{equation}\label{eq:bm:gibbs_sampling}
    z_{i'} = \begin{cases}
        1,\quad &\textrm{ if } \sigma\left(\theta_{i'} + \sum_{i \neq {i'}} W_{i',i} z_i\right) \geq \frac{1}{2},\\
        0,\quad &\textrm{ otherwise},
    \end{cases}
\end{equation}
where $\theta_i$ denotes the $i$-th element of $\thetavec$, $W_{i',i}$ denotes the element $(i',i)$ of $\W$ capturing the correlation between $Z_i$ and $Z_{i'}$, and $\sigma(x) = \frac{1}{1+\mathrm{e}^{-x}}$, and $z_i \in \{ 0, 1 \}$, $i = 1,\ldots,m$ correspond to a specific state of the \ac{BM}, i.e., one realization of $\Z$.

Specifically, when $m$ is large and \eqref{eq:bm:gibbs_sampling} is executed very often, the empirical statistics of $\lbrace Z_i \rbrace_1^m = \lbrace Z_i | i = 1,\ldots,m \rbrace$ converge to
\vspace*{-0.4cm}
\begin{equation}
    p_{\Z}(\z) = \frac{1}{\partition} \exp\left(\frac{1}{2} \z^\transpose \W \z + \z^\transpose \thetavec \right),
\end{equation}
where $\partition$ is a normalization constant, and $\thetavec$ and $\W$ can be selected such that, after identifying $\Z'$ with a subset $\lbrace \tilde{Z}_j\rbrace_1^n$ of $\Z$, the empirical statistics of $\lbrace \tilde{Z}_j\rbrace_1^n  \subseteq \{ Z_i \}_1^m$ converge to $q_{\Z'}(\z')$.
{\em Training} the \ac{BM} then refers to updating $\thetavec$ and $\W$ iteratively according to the difference between the observed statistics of $\lbrace \tilde{Z}_j\rbrace_1^n$ and $q_{\Z'}(\z')$ as described in \cite[Eq. (4)]{poole2017_chemical_BMs}. While, for general $q_{\Z'}(\z')$, $m \gg n$ and the consideration of higher-order moments of $q_{\Z'}(\z')$ may be required in order to achieve high approximation quality, some $q_{\Z'}(\z')$ can be approximated by \acp{BM} already with $m=n$ that are trained using only the first- and second-order moments $\expectation{q}{\Z'}$ and $\expectation{q}{\Z'\Z'^\transpose}$ of $q_{\Z'}(\z')$.
In the following, we will restrict our considerations to this special case and show in the following section that this is indeed sufficient for accurate detection in the considered system model.
Since a detailed account of \acp{BM} is beyond the scope of this paper, we refer the interested reader to \cite{mackay2003information} for further information on the topic.

\subsection{Detection with Boltzmann Machines}
\label{sec:BM_detection}
In order to use \acp{BM} for detection in the context of the system model introduced in Section~\ref{sec:system_model}, we set $\Z = \begin{bmatrix} \Xhat & \Y \end{bmatrix}^{\transpose}$, where we recall that $\Xhat$ denotes a surrogate \ac{RV} representing the estimate of the transmitted symbol $X$ and $\Y$ denotes the vector of \acp{RV} modeling the \ac{RX} receptor states.
The following theorem guarantees that the \ac{BM} can be designed to possess the MAP property in agreement with Definition~\ref{def:system_model:map_property}.

\vspace*{-3mm}
\begin{theorem}\label{thm:bm:bm_map_detection}
If the optimal threshold $\numap$ is known, parameters $\W$ and $\thetavec$ can be chosen such that the \ac{BM} defined by $(\Z, \W, \thetavec)$ possesses the \ac{MAP} property.\label{th:FVBM_MAP_property}
\end{theorem}
\vspace*{-5mm}
\begin{proof}
The theorem follows after setting $W_{i,i'}=\wxy$ for all $i\neq i'$ and $\theta_i = -(\MAPthreshold-\frac{1}{2})\wxy$ for some positive constant $\wxy$.
For the full proof, we refer the reader to \cite{bahe_nanocom}.
\end{proof}
\vspace*{-2mm}
Theorem~\ref{thm:bm:bm_map_detection} not only tells us that it is possible to realize \ac{MAP} detection using \acp{BM}, it also reveals how to properly select $\W$ and $\thetavec$ for the system model considered in this paper.
However, since \acp{BM} can learn arbitrary joint probability distributions of binary \acp{RV}, they can in general be utilized for detection even if the conditions for Theorem~\ref{thm:bm:bm_map_detection} are not fulfilled.

Certainly, one advantage of data-driven detectors (like the one we introduced in this section) as compared to purely model-based detectors is that the former ones can be \textit{trained} from measurements or simulations, while the latter ones require {\em a priori} knowledge and/or strong assumptions on the physical system parameters.
For example, even if the physical channel geometry does not lend itself to analytical modeling, it may still be possible to perform particle-based simulations and use the obtained data to train a data-driven detector. Also, if a testbed is available, measurement data might be used for training the detector.
Trying to leverage these advantages, several neural network-based detectors utilizing digital hardware implementations have already been proposed in the \ac{MC} literature \cite{wei_parzen_pnn,qian_receiver_anns,farsad_neural_networks}.
Now, when it comes to implementing machine learning-based detectors in cellular environments, it has been shown that \ac{CRN}-based implementations of feed-forward neural networks can in principle be constructed \cite{anderson_neural_networks}.
However, the performance of the \ac{CRN} implementations proposed in \cite{anderson_neural_networks} might be subject to significant performance degradation in real world scenarios since they are based on the assumption of deterministic molecule counts, while stochastic fluctuations in single cell environments can be expected, cf.~Section~\ref{sec:system_model:crn-based_detection}.
In contrast, in the following section, we present a \ac{CRN}-based implementation of the \ac{BM}-based detector introduced in this section that not only takes into account the randomness of molecule counts but actually exploits it.
\vspace*{-5mm}
\subsection{Implementing Boltzmann Machines via CRNs}
In \cite{poole2017_chemical_BMs}, it was demonstrated how, in general, \acp{BM} can be implemented using \acp{CRN}.
To this end and analogous to the definition of mapping function $f$ in Section~\ref{sec:system_model}, every \ac{RV} $Z_i$ is represented by molecule species $\Zon{i}$ and $\Zoff{i}$.
Specifically, $n_{\Zon{i}}=1$ and $n_{\Zoff{i}}=1$ are equivalent to $Z_i=1$ and $Z_i=0$, respectively, and $n_{\Zon{i}} + n_{\Zoff{i}}=1$.
The resulting \ac{CRN} then comprises the molecule species $\{\Zon{1}, \Zoff{1}, \dots, \Zon{m}, \Zoff{m}\}$.
For this specific construction, it has been shown in \cite{poole2017_chemical_BMs} that using a method called \textit{Taylor mapping} the steady state of the \ac{CRN} can be matched quite accurately to the steady state distribution of the corresponding \ac{BM} by defining appropriate chemical reactions. 

For the specific case considered in this paper, by specializing the Taylor mapping from \cite{poole2017_chemical_BMs} to the \ac{BM}-based detector derived in the previous section and applying the mappings $f$, $g$ introduced in Section~\ref{sec:system_model:crn-based_detection}, we obtain the set of chemical species
\begin{equation}
    \mathcal{S}_{\mathrm{TM},\Xhat} = \{\Xhaton,\Xhatoff,\Yon{1},\Yoff{1},\dots,\Yon{\nr},\Yoff{\nr}\}.
\end{equation}
Then, we recall that the task of the \ac{CRN}-based detector is, for any fixed $\y$, to generate samples of $\hat{X}$ according to $p_{\hat{X}|\Y}(\hat{X}|\Y=\y)$.
In terms of molecule counts, this observation implies that the species counts of all species $\Yon{i},\Yoff{i}$ are fixed for one detection interval\footnote{See \cite{poole_clamping} for a discussion of different methods to {\em clamp} chemical species, i.e., to fix their molecule counts.}.
Hence, we need to consider only the subset of chemical reactions from the Taylor mapping from \cite{poole2017_chemical_BMs} that possibly change the molecule count of $\Xhaton$ and/or $\Xhatoff$.
This leads to the following chemical reactions.
\begin{equation}
    \begin{split}
        \Xhatoff &\xrightleftharpoons[k\cdot(1+|\theta_1|)]{k} \Xhaton, \\
        \Yon{1} + \Xhatoff &\xrightarrow{k \cdot W_{1,2}} \Yon{1} + \Xhaton, \\[-1em]
        &\;\;\;\;\;\vdots \\
        \Yon{\nr} + \Xhatoff &\xrightarrow{k \cdot W_{1,\nr+1}} \Yon{\nr} + \Xhaton\label{eq:BM_CRN_implementation}.
    \end{split}
\end{equation}
The following theorem confirms that \eqref{eq:BM_CRN_implementation} is indeed a sensible choice and the \ac{MAP} property of the underlying \ac{BM} is preserved in the corresponding \ac{CRN}.

\begin{theorem}
Consider the \ac{CRN} $\mathcal{C}_{\mathrm{TM},\Xhat}=(\mathcal{S}_{\mathrm{TM},\Xhat}, \mathcal{R}_{\mathrm{TM},\Xhat},\kvec_{\mathrm{TM},\Xhat})$, where $\mathcal{R}_{\mathrm{TM},\Xhat}$ and $\kvec_{\mathrm{TM},\Xhat}$ comprise the chemical reactions and rate constants defined in \eqref{eq:BM_CRN_implementation}, respectively. Then, $\mathcal{C}_{\mathrm{TM},\Xhat}$ preserves the \ac{MAP} property if $\W$ and $\thetavec$ are chosen according to Theorem~\ref{thm:bm:bm_map_detection}, i.e., if the \ac{BM} has the \ac{MAP} property.\label{th:bm_crn_map}
\end{theorem}
\vspace*{-5mm}
\begin{proof}
    For the proof of this theorem, we refer the reader to \cite{bahe_nanocom}.
\end{proof}
\vspace*{-7mm}

\subsection{Chemical Decision Making}\label{sec:bm:decision_making}
\vspace*{-2mm}
The \ac{CRN} introduced in the previous section generates samples according to $\Pr\left[\Xhat|\Y\right]$ in agreement with Definition~\ref{def:system_model:map_property}.
In order to convert these samples into detection decisions, the following two steps are applied.
\begin{enumerate}
	\item \Ac{CRN} $\mathcal{C}_{\mathrm{TM},\Xhat}$ is extended to record the empirical frequencies at which $\Xhaton$ and $\Xhatoff$ are observed. 
    To this end, we introduce two counting species $\XhatCon$ and $\XhatCoff$ that are created in the presence of $\Xhaton$ and $\Xhatoff$, respectively, with the same rate constant $\rrc{c}$ according to the following reactions
    \begin{equation}\label{eq:bm:counting_reactions}
        \Xhaton \xrightarrow{\rrc{c}} \Xhaton + \XhatCon,\;\;\;\;\;\;\Xhatoff \xrightarrow{\rrc{c}} \Xhatoff + \XhatCoff.
    \end{equation}
		\item The following additional reaction is introduced such that $\XhatCon$ and $\XhatCoff$ annihilate each other with rate constant $\rrc{a}$ similar to \cite{bi_crn_microfluidic_tx_rx, riaz_map_spatially_partioned},
	\begin{equation}
		\XhatCon + \XhatCoff \xrightarrow{\rrc{a}} \emptyset.\label{eq:BM_annilihation_reaction}
	\end{equation}
\end{enumerate} 
For $\frac{\rrc{a}}{\rrc{c}} \gg 1$ and $\rrc{a},\rrc{c}\rightarrow\infty$, only either $\XhatCon$ or $\XhatCoff$ molecules are observed (but not both at the same time) depending on whether $\Pr[\Xhat=1|\Nrb=\nrb] \geq 1/2$ or not; hence, the presence of $\XhatCon$ corresponds to $\Xhat=1$ and the presence of $\XhatCoff$ to $\Xhat=0$.
For $\rrc{a}$, $\rrc{c}$ finite, $\XhatCon$ and $\XhatCoff$ molecules may coexist; in this case the detection decision is probabilistic and the probability of $\Xhat=1$ ($\Xhat=0$) corresponds to the probability of observing $\XhatCon$ ($\XhatCoff$) relative to observing $\XhatCoff$ ($\XhatCon$).
See also the definition of the \ac{BER} in Section~\ref{sec:evaluation:adaptive_detector:ber}.

Since \ac{BM} models are based on binary \acp{RV}, it is natural to model the binary state (bound vs.~unbound) of each receptor individually in the context of \acp{BM}.
However, as a consequence of this modeling decision, the numbers of chemical species and reactions of $\crn{TM,\hat{X}}$ as defined in Theorem~\ref{th:bm_crn_map} scale linearly with the number of receptors $\nr$.
In the following section, we consider an alternative to \ac{BM}-based \ac{CRN} design, which circumvents this scaling of the \ac{CRN} complexity.

In summary, in this section, we showed how an offline-trained \ac{BM} can be implemented chemically and how its stochastic output can be converted into quasi-deterministic detection decisions for further processing.

\section{Adaptive Detector Implementation}
\label{sec:adaptiveRX}
In the previous section, we showed how \ac{CRN}-based chemical detectors can be derived from a \ac{BM} model.
These \ac{BM}-based detectors are especially well-suited for detection in scenarios in which analytical channel models are not available, but measurements or simulations of the channel can be performed before the \ac{RX} is deployed.
However, such prior knowledge on the channel is often not available or the channel characteristics, like the background noise level, might change over time.
For detection in these cases, in this section, we propose a \ac{CRN}-based detector that can be trained using pilot symbols after its deployment, i.e., online, and in this way adapt to \textit{a priori} unknown and/or time-varying channels.
To this end, we first propose a novel \ac{CRN} design for detection that has significantly lower complexity in terms of the required chemical species and reactions compared to the \ac{BM}-based \ac{CRN} presented in the previous section, but still preserves the \ac{MAP} property.
Then, we extend the proposed low-complexity \ac{CRN} by a chemical learning rule that allows online training using pilot symbols.

\subsection{Low Complexity CRN-based MAP Detection}\label{sec:adaptive_decision_rule}
We recall from Section~\ref{sec:system_model:crn-based_detection} that $\Pr\left[X| \Y\right] = \Pr\left[X| \Nrb \right]$, i.e., the number of bound receptors provides sufficient information for detection.
Following this observation, the input to the \ac{CRN}-based detector proposed in this section is represented as chemical species $\YN$, where $n_{\YN}=\Nrb$, i.e., the molecule count of $\YN$ is fixed to $\Nrb$ for each detection interval.

Based on this input definition, we further define the following set of reactions $\mathcal{R}_{\mathrm{LC}}$
\begin{subequations}\label{eq:low_complexity_detector}
\begin{align}
	\YN + \Xhatoff &\xrightarrow{\kon} \YN + \Xhaton \label{eq:low_complexity_detector_yon}, \\
	\Won + \Xhaton &\xrightarrow{\koff} \Won + \Xhatoff \label{eq:low_complexity_detector_won},
\end{align}
\end{subequations}
with reaction rate constants $\kvec_{\mathrm{LC}}=[\kon,\koff]^\transpose$.
Here, $\Won$ is a new species of weight molecules that deactivates $\Xhaton$ and serves as the threshold for detection.
The intuition behind \eqref{eq:low_complexity_detector} is as follows: The more receptors are bound, the more likely reaction \eqref{eq:low_complexity_detector_yon} occurs, i.e., the more likely $\Xhatoff$ is converted to $\Xhaton$, cf.~\eqref{eq:mass_action_example}. On the other hand, the more weight molecules exist, the more likely $\Xhaton$ is converted to $\Xhatoff$. Hence, for a given number of $\Won$ molecules, $\Xhaton$ is more often observed than $\Xhatoff$ if a critical number of bound receptors is surpassed. Below that number, $\Xhatoff$ is more often observed than $\Xhaton$ in this symbol interval.
In the following theorem, we will see that \eqref{eq:low_complexity_detector} is indeed a sensible choice with respect to the detection performance of the resulting \ac{CRN}.

\begin{theorem}\label{thm:adaptive_crn:MAP_property}
Consider the \ac{CRN} $\crn{LC}=\left(\mathcal{S}_{\mathrm{LC}}, \mathcal{R}_{\mathrm{LC}}, \kvec_{\mathrm{LC}}\right)$, where $\mathcal{S}_{\mathrm{LC}} = \{\YN, \Won, \Xhatoff, \Xhaton\}$. For an appropriate and fixed number of weight molecules, $\nwa$, $\crn{LC}$ possesses the \ac{MAP} property.\label{th:adaptive_map}
\end{theorem}
\begin{proof}
    We take the same detailed balance-based approach as in the proof of Theorem~\ref{th:bm_crn_map}, cf.~\cite{bahe_nanocom}, to compute the steady-state probability of observing $\Xhaton$.
    Here, the rate of converting $\Xhatoff$ to $\Xhaton$ is given by $n_{\YN} \cdot \rrc{on}$ if $n_{\Xhatoff}=1$ and $0$ otherwise. $\Xhaton$ is converted to $\Xhatoff$ with rate $n_{\mol{W}} \cdot \rrc{off}$ if $n_{\Xhaton}=1$ and rate $0$ otherwise.
    Hence, we obtain
    \begin{equation}
        \Pr[n_{\Xhaton}=1|n_{\YN}=\nrb] = \frac{\nrb}{\nrb+\frac{\rrc{off}}{\rrc{on}}\nwa} \label{eq:adaptive_rx_theorem_2_Pr}.
    \end{equation}
    Solving $\Pr[n_{\Xhaton}=1|\Nrb=\nrb] \geq 0.5$ for $\nrb$ yields $\nrb \geq \frac{\rrc{off}}{\rrc{on}}\nwa$.
    For example, for $\frac{\rrc{off}}{\rrc{on}}=1$ and large $\rrc{on}$, the conditions for the \ac{MAP} property are fulfilled if $\nwa=\MAPthreshold$.
\end{proof}

We conclude from Theorem~\ref{thm:adaptive_crn:MAP_property} that \ac{CRN}-based \ac{MAP} detection can be achieved by $\crn{LC}$ with only a few different chemical species, namely four, and a small set of chemical reactions.

\subsection{Chemical Decision Making}
\label{sec:adaptive_chemical_decisions}
In principle, the chemical decision making process discussed for $\crn{\mathrm{TM},\Xhat}$ in Section~\ref{sec:bm:decision_making} can analogously be applied to $\crn{LC}$.
However, the decision making can also be integrated into $\crn{LC}$ by modifying \eqref{eq:low_complexity_detector} to directly produce counting species $\XhatCoff$ and $\XhatCon$.
Together with the annihilation reaction, this yields the following set of reactions $\mathcal{R}_{\mathrm{LC}'}$
\begin{subequations}\label{eq:adaptive_decision}
\begin{align}
	\YN &\xrightarrow{\kon} \YN + \XhatCon \label{eq:adaptive_decision_1},\\
	\Won  &\xrightarrow{\koff} \Won + \XhatCoff \label{eq:adaptive_decision_2},\\
    \XhatCon + \XhatCoff &\xrightarrow{\rrc{deg}} \emptyset \label{eq:adaptive_decision_3}.
\end{align}
\end{subequations}

As can be easily checked, $\crn{LC'}=\left(\mathcal{S}_{\mathrm{LC}'}, \mathcal{R}_{\mathrm{LC}'}, \kvec_{\mathrm{LC}'}\right)$, where $\mathcal{S}_{\mathrm{LC}'} = \lbrace \YN, \Won, \XhatCon, \XhatCoff \rbrace$ and $\kvec_{\mathrm{LC}'} = [\kon, \koff, \rrc{deg}]$, preserves the MAP property of $\crn{LC}$ if $\nwa\!=\!\MAPthreshold$, $\rrc{deg}\gg\kon\!=\!\koff$, for sufficiently large $\kon,\koff$, since in that case $\Pr\left[n_{\XhatCon} > 0 | n_{\YN}=\nrb\right] \geq 1/2$ \ac{iff} $\nrb \geq \MAPthreshold$.

\subsection{Pilot Symbol-based Learning Rule}\label{sec:in_silico_learning}
When the optimal threshold $\MAPthreshold$ is unknown or changes over time, it can be estimated based on the transmission of pilot symbols. To facilitate a chemical implementation of the threshold estimation and adoption by the \ac{RX}, we propose a very simple learning algorithm: After reception of a pilot symbol in symbol interval $l$, the symbol estimate $\xhat$ is compared to the true pilot symbol value $\xtrue$ and the number of $\Won$ molecules $\nwa$ is adapted according to 
\begin{equation}
     \nwa[l+1] = \begin{cases}
                    \nwa[l], & \text{\quad if}\;\; \xhat=\xtrue, \\
                    \nwa[l]+1, & \text{\quad if}\;\; \xhat=1 \text{ and } \xtrue = 0, \\
                    \nwa[l]-1, & \text{\quad if}\;\; \xhat=0 \text{ and } \xtrue = 1,
                \end{cases}\label{eq:update_rule}
\end{equation}
i.e., if $\xhat=1$ and $\xtrue=0$ ($\xhat=0$ and $\xtrue=1$), the number of weight molecules $\Won$ is increased (decreased) such that the rate of \eqref{eq:adaptive_decision_2} increases (decreases) and more (fewer) $\XhatCoff$ molecules are produced.
Clearly, \eqref{eq:update_rule} is suboptimal in the sense that threshold updates may even be performed in case of a wrong decision if the optimal threshold $\MAPthreshold$ is already reached, since the detector is unaware of whether $\nwa[l]$ is equal to $\MAPthreshold$ or not.
However, we will see in Section~\ref{sec:simulations} that the performance loss due to this suboptimality is relatively small, while we confirm in the following section that the complexity of the chemical implementation of \eqref{eq:update_rule} is indeed very low.

\subsection{Chemical Learning}
\label{sec:adaptive_RX_chemical_learning}
The pilot symbol-based learning rule in \eqref{eq:update_rule} can be chemically implemented by the following chemical reactions
\vspace*{-5mm}
\begin{subequations}\label{eq:adaptive_learning_reaction}
\begin{align}
    \XhatCon + \XtrueOff + \Hmol &\xrightarrow{\rrc{l,1}} \XhatCon + \XtrueOff + \Won \label{eq:adaptive_learning_reaction_1}, \\
    \XhatCoff + \XtrueOn + \Won + \Hmol &\xrightarrow{\rrc{l,2}} \XhatCoff + \XtrueOn,\label{eq:adaptive_learning_reaction_2}
\end{align}
\end{subequations}
with reaction rate constants $\rrc{l,1}$ and $\rrc{l,2}$. Here, the true pilot symbol is encoded by chemical species $\XtrueOn$ and $\XtrueOff$ for $\xtrue=1$ and $\xtrue=0$, respectively. $\Hmol$ is a helper molecule that is consumed during the update step in order to ensure that exactly one $\Won$ molecule is added or removed. In Section~\ref{sec:timing_mechanism_helper_molecules}, we will explain how $\Hmol$, $\XtrueOn$, and $\XtrueOff$ can be produced in practice. 
\section{Timing Mechanism for a Complete Receiver Implementation}
\label{sec:timing}
\begin{figure*}
    \vspace*{-12mm}
    \centering
    \includegraphics[width=0.88\textwidth]{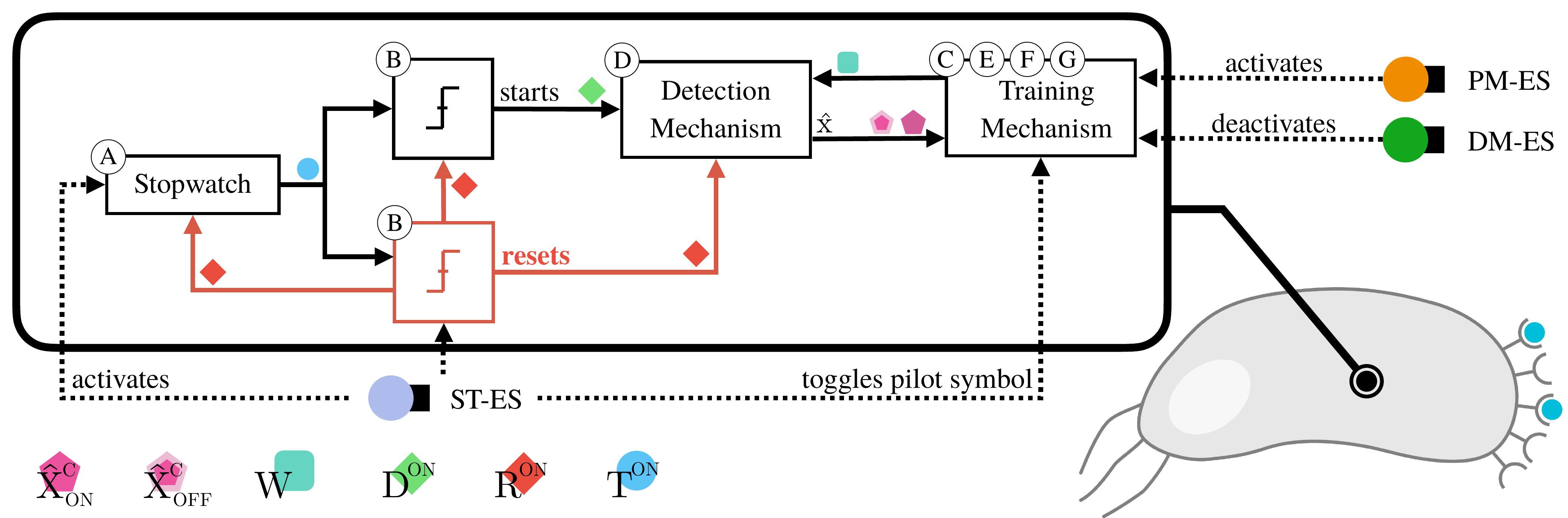}
    \vspace*{-3mm}
    \caption{\textbf{Main functionalities of the \ac{RX} components.} \ac{STES} ends the previous reset phase, starts the stopwatch, and updates the stored value of the pilot symbol. Once the stopwatch has run long enough, the detection phase starts and the symbol value is estimated. If the training mechanism is active, it updates $\nwa$ if necessary. Eventually, the reset phase starts and the \ac{RX} is reset. The letters on the individual blocks indicate where the chemical implementation is shown in Figure~\ref{fig:system_overview}.}\label{fig:flowchart}
    \vspace*{-10mm}
\end{figure*}
\begin{figure*}[t]
    \vspace*{-10mm}
    \centering
    \includegraphics[width=\textwidth]{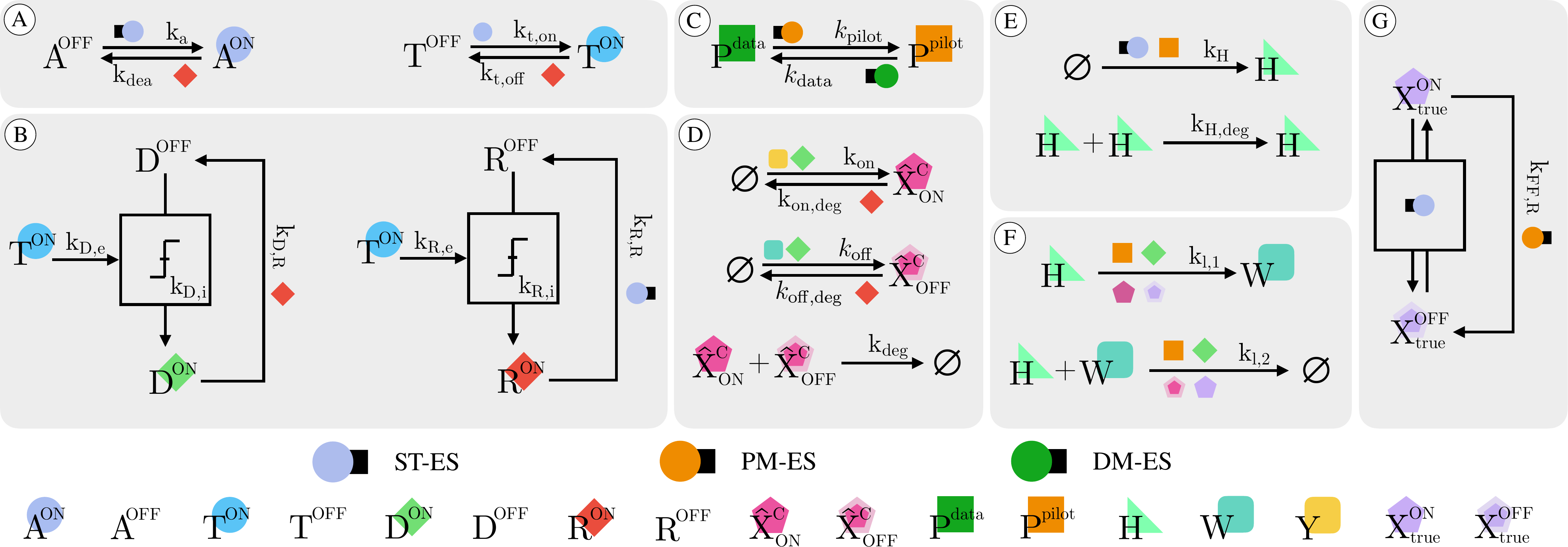}
    \vspace*{-8mm}
    \caption{\textbf{Overview of the adaptive \ac{RX}.} A: Stopwatch. B: Switches to define computation phases (cf. Figure~\ref{fig:chemical_switch}). C: \ac{PMES} and \ac{DMES} prompt the transition between pilot and data phases. D: Detection mechanism. E: Production of the helper molecules required for training. F: Training process. G: Chemical flip-flop storing the true value of pilot symbols (cf. Figure~\ref{fig:chemical_flipflop}). The notation of the chemical reactions in this figure slightly extends \eqref{eq:system_model:reaction_external_catalyst} such that all types of catalysts (external and non-external) are here depicted by their corresponding symbols above or below the arrows, e.g., the upper reaction in block E is equivalent to $\Ptrain \xrightarrow{\mathrm{ST-ES}\,\rrc{H}} \Hmol + \Ptrain$. The symbol $\varnothing$ denotes any chemical species that is abundant and of no further relevance to the considered \ac{CRN}.}
    \label{fig:system_overview}
    \vspace*{-10mm}
\end{figure*}

In order to guarantee that the chemical reactions during the detection process as defined in Sections~\ref{sec:BMbasedRX} and \ref{sec:adaptiveRX} happen at the right time and under correct initial conditions, our cellular \ac{RX} needs some additional intracellular timing mechanisms.
In case of the adaptive detection scheme presented in Section~\ref{sec:adaptiveRX}, these mechanisms should also coordinate the respective detection and training processes.
Since it requires more coordination effort as compared to the \ac{BM}-based detector, we develop the timing mechanism for the adaptive detector implementation from Section~\ref{sec:adaptiveRX}.
However, the proposed design can easily be extended to the chemical \ac{BM}-based detector by omitting the training mechanism and integrating the corresponding chemical detection mechanism from Section~\ref{sec:BMbasedRX}.

Before we develop the chemical implementation of the proposed timing mechanism, we outline the required functional components and their interactions on a schematic level.
As illustrated in Figure~\ref{fig:flowchart}, at the beginning of each symbol interval, the external signal \ac{STES} activates a chemical {\em stopwatch} (see Section~\ref{sec:timing:stopwatch}) and {\em toggles} the pilot symbol value in the training mechanism (see Section~\ref{sec:flip_flop}).
The purpose of the stopwatch is to (i) first initiate the detection phase and (ii) later initiate the reset phase via two {\em chemical switches} (Section~\ref{sec:timing:chemical_switch}).
During the detection phase, the \ac{CRN}-based detector $\crn{LC'}$, as developed in Section~\ref{sec:adaptiveRX}, generates an estimate of the transmit signal, while, controlled by two external signals, \ac{PMES} and \ac{DMES}, the \ac{RX} is either in {\em pilot mode}, i.e., the training mechanism is active, or in {\em data mode}, i.e., the training mechanism not active (see Section~\ref{sec:timing:switching_pm_dm}).
When the training mechanism is active, it utilizes the transmit symbol estimate from the detector and the stored pilot symbol value to update the number of weight molecules according to the chemical learning rule introduced in Section~\ref{sec:adaptive_RX_chemical_learning} (see Section~\ref{sec:timing_mechanism_helper_molecules}).
Finally, during the reset phase, some {\em chemical cleanup} tasks are performed that reset the state of the \ac{RX} in preparation for the next symbol interval (see Section~\ref{sec:timing:cleanup}).
The reset phase of the current symbol interval is terminated by the \ac{STES} initiating the next symbol interval.
In the following, we discuss how to chemically implement the components that are required to realize all these functionalities in the considered cellular \ac{RX}.

\subsection{Stopwatch}\label{sec:timing:stopwatch}
The key component of our timing mechanism is the chemical stopwatch illustrated in Figure~\ref{fig:system_overview}-A.
In the presence of $\ac{STES}$, i.e., at the beginning of each symbol interval, $\Aoff$ molecules are converted to $\Aon$ molecules, where the presence of the $\Aon$ molecules indicates that the stopwatch is {\em activated}.
The rate of conversion of $\Aoff$ to $\Aon$, $\rrc{a}$, is relatively fast, such that with high probability all existing $\Aoff$ molecules are converted to $\Aon$ at the beginning of each symbol interval.
Now, in its activated state, due to the presence of $\Aon$ molecules, the chemical stopwatch produces $\Ton$ molecules from $\Toff$ molecules with rate constant $\rrc{t,on}$.
In this way, the number of $\Ton$ molecules increases as the time since the beginning of the current symbol interval elapses.
The chemical switches presented in the following subsection leverage this information to trigger the time-dependent execution of specific other chemical reactions.

During the reset phase, i.e., in the presence of the $\Ron$ molecules produced by the reset mechanism as discussed later in this section, the chemical stopwatch is deactivated by converting the $\Aon$ molecules back to $\Aoff$ molecules with reaction rate constant $\rrc{dea}$ and the existing $\Ton$ molecules are converted back to $\Toff$ molecules with reaction rate constant $\rrc{t,off}$. 

\subsection{Chemical Switches}\label{sec:timing:chemical_switch}
\begin{figure}
    \vspace*{-12mm}
    \begin{minipage}{0.48\textwidth}
        \centering
        \includegraphics[width=3.0in]{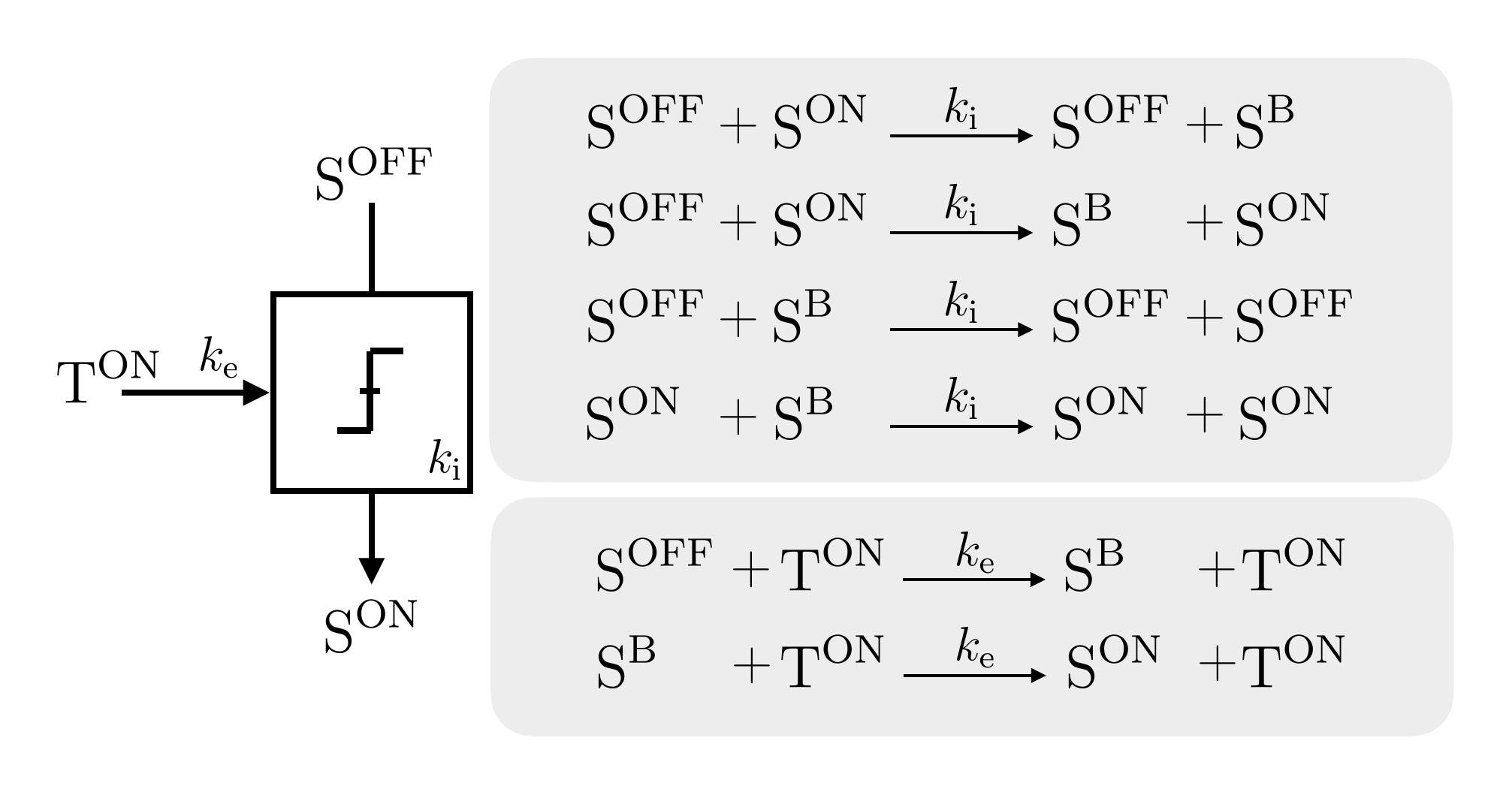}
        \vspace*{10mm}
        \caption{\textbf{Chemical Switch.} The reactions in the upper box implement a bi-stable behavior that can be altered by the $\Ton$ molecules via the reactions in the lower box \cite{scirep_approx_majority}.}
        \label{fig:chemical_switch}
    \end{minipage}
    \hfill
    \begin{minipage}{0.48\textwidth}
        \centering
        \includegraphics[width=3.0in]{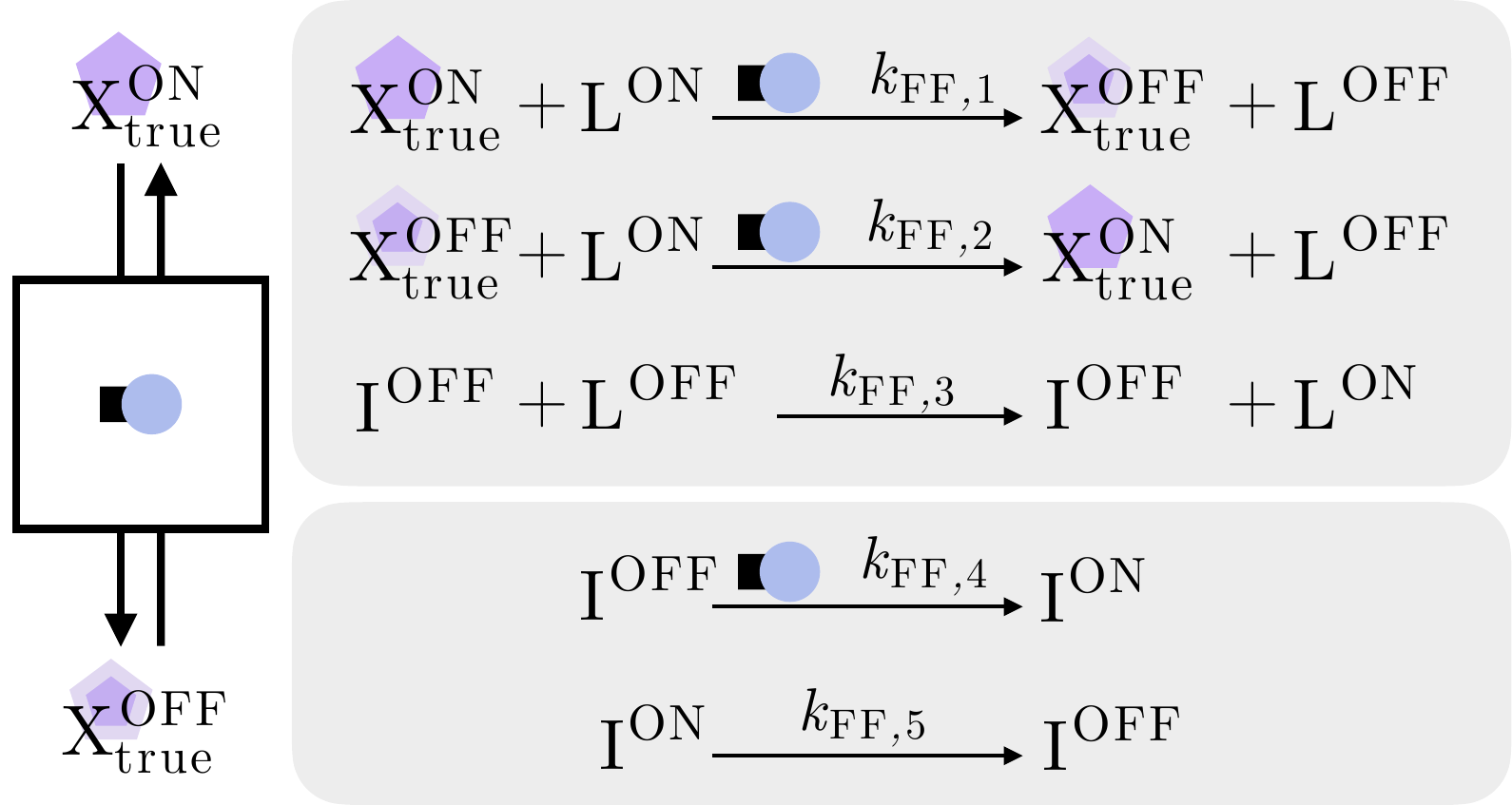}
        \includegraphics[width=3.2in]{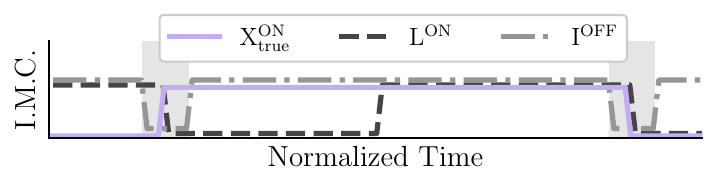}
        \vspace*{-5mm}
        \caption{\textbf{Chemical Flip-Flop.} \ac{STES} (top: blue lamp, bottom: grey shaded) can catalyze the flip of the molecule indicating the true value of the pilot symbol. The evolution of the Idealized Molecule Count (I.M.C.) is shown at the bottom.}
    \label{fig:chemical_flipflop}
    \end{minipage}
    \vspace*{-10mm}
\end{figure}
As indicated in Figure~\ref{fig:flowchart}, the timing molecules $\Ton$ produced by the stopwatch are utilized by two \textit{chemical switches}, cf. Figure~\ref{fig:system_overview}-B, to switch on the detection and reset phases.
Specifically, one chemical switch, the {\em detection switch}, initiates the detection phase by converting $\Doff$ molecules to $\Don$ molecules once the number of $\Ton$ molecules exceeds some threshold.
Similarly, the second chemical switch, the {\em reset phase switch}, initiates the reset phase by converting $\Roff$ molecules to $\Ron$ molecules once the number of $\Ton$ molecules exceeds the respective threshold.
Hereby, the specific threshold values depend on the respective reaction rate constants of the chemical switches and the molecule counts of the chemical species $\Doff$, $\Don$ and $\Roff$, $\Ron$, respectively.

The chemical implementation of these switches, which is adopted from \cite{scirep_approx_majority}, is illustrated in Fig.~\ref{fig:chemical_switch}.
In this illustration, the notation $\Soff$, $\Sb$, and $\Son$ represents either $\Doff$, $\Db$, and $\Don$ or $\Roff$, $\Rb$, and $\Ron$ depending on which of the aforementioned switches is considered.
From Fig.~\ref{fig:chemical_switch}, we observe that the switch can be divided into two functionally distinct units.
One unit (upper box in Fig.~\ref{fig:chemical_switch}) computes an \textit{approximate majority vote} between $\Son$ and $\Soff$ molecules, i.e., if there exist more $\Son$ than $\Soff$ molecules, the chemical reactions ensure that only $\Son$ molecules are left after some time and vice versa.
In other words, this unit renders the chemical switches {\em bistable}, i.e., they possess two stable states in which either only $\Soff$ or only $\Son$ molecules exist.
We refer to these two states as OFF state and ON state, respectively.
The intermediate species $\Sb$ is required for the realization of two non-linear positive feedback loops, on which the switches' bistability depends.

\textbf{OFF$\rightarrow$ON:} In order to transition from OFF state to ON state, the second unit (chemical reactions in the lower box in Fig.~\ref{fig:chemical_switch}) converts $\Soff$ molecules to $\Son$ molecules, whereby the rate of the conversion depends on the number of $\Ton$ molecules; the more $\Ton$ molecules exist, the faster happens the conversion.
Assuming the switch is in OFF state initially and the number of $\Ton$ molecules increases monotonically during one symbol interval, then, eventually, the switch transitions to its ON state as the number of $\Ton$ molecules surpasses its inherent threshold.

\textbf{ON$\rightarrow$OFF:} The detection switch is reset to its OFF state during the reset phase, i.e., when $\Ton$ molecules are being removed from the system and the presence of $\Ron$ molecules catalyzes the conversion of $\Don$ molecules to $\Doff$ molecules (see Fig.~\ref{fig:system_overview}-B).
Finally, the reset phase switch is reset to its OFF state in the presence of \ac{STES}, since \ac{STES} catalyzes the conversion of $\Ron$ molecules to $\Roff$ molecules.

The output molecules produced by the considered chemical switches, $\Don$ and $\Ron$, are utilized inside the \ac{RX} to activate certain chemical reactions exclusively in the context of the detection phase and reset phase, respectively.
As shown in Figure~\ref{fig:system_overview}-D, the detection reactions \eqref{eq:adaptive_decision} are extended in the fully chemical \ac{RX} implementation presented in this section to require the presence of $\Don$ molecules such that the reactions occur only during the detection phase.
Also, the training reactions \eqref{eq:adaptive_learning_reaction} require the presence of $\Don$ molecules in the fully chemical \ac{RX} implementation. 
We will discuss in Section~\ref{sec:timing:cleanup} which specific reactions in the context of the reset phase are activated by the presence of $\Ron$ molecules.

\vspace*{-5mm}
\subsection{Additional Reactions for Training}\label{sec:timing:training}
There are two issues that arise when embedding the chemical pilot symbol-based training mechanism developed in Section~\ref{sec:adaptive_RX_chemical_learning} into a cellular \ac{RX} that processes not only one but multiple symbols.
First, in each training cycle exactly one helper molecule $\mol{H}$ needs to be available to catalyze a change in the number of weight molecules $\mol{W}$ by at most one molecule.
Second, the current pilot symbol value represented by the molecules $\XtrueOn$ and $\XtrueOff$ needs to be toggled depending on the transmitted pilot symbol.

\subsubsection{Production of Helper Molecules}
\label{sec:timing_mechanism_helper_molecules}
To ensure that the chemical implementation of learning rule \eqref{eq:update_rule} changes the number of $\Won$ molecules by not more than one molecule per pilot symbol, exactly one helper molecule $\Hmol$, which can be consumed by the update reactions \eqref{eq:adaptive_learning_reaction}, should be available at each pilot symbol transmission.
To satisfy this constraint, we define the two chemical reactions illustrated in \mbox{Figure~\ref{fig:system_overview}-E}.
The first reaction produces helper molecules $\Hmol$ with reaction rate constant $\rrc{H}$ if \ac{STES} is present (which is the case at the beginning of each symbol interval) {\em and} the system is in pilot mode, i.e., in the presence of $\Ptrain$ molecules, cf.~Section~\ref{sec:timing:switching_pm_dm}.
The second reaction in Figure~\ref{fig:system_overview}-E ensures that one of two $\Hmol$ molecules is annihilated with rate constant $\rrc{H,deg}$ if two $\Hmol$ molecules react. This ensures that no more than one molecule is available if $\rrc{H,deg} \gg \rrc{H}$.

\subsubsection{A Chemical Flip-Flop for Pilot Symbols}
\label{sec:flip_flop}
In Section~\ref{sec:adaptive_RX_chemical_learning}, we did not discuss how the \ac{RX} obtains the true value of the pilot symbol.
While storing this information can be a complex task in general, we may assume here without loss of generality (since there exists no \ac{ISI}) that the pilot symbol sequence is of the form $\langle0,1,0,1,\dots\rangle$, i.e., $1$s and $0$s alternate.
In this case, a \textit{chemical flip-flop} can be utilized to retrieve the true pilot symbol value \mbox{for each symbol interval}.

Figure~\ref{fig:chemical_flipflop} shows one possible implementation of a chemical flip-flop (top) and an idealized time evolution of its relevant molecule counts (bottom).
In the flip-flop, there exists always either one $\XtrueOn$ or one $\XtrueOff$ molecule. 
Besides the molecule species $\XtrueOn$ and $\XtrueOff$, the flip-flop requires two additional molecules that can each be either in activated or deactivated state and are denoted by $\Lon$, $\Loff$ and $\Ion$, $\Ioff$, respectively.
Now, upon observing \ac{STES}, $\XtrueOn$ is switched to $\XtrueOff$ and vice versa with rate constants $\rrc{FF,1}$ and $\rrc{FF,2}$, respectively.
At the same time, it is ensured that only one such conversion occurs while \ac{STES} is active, since the respective reactions also convert $\Lon$ to $\Loff$ and only one $\Lon$ molecule is available. 
Finally, in order to reset this mechanism, $\Loff$ is converted back to $\Lon$ with rate constant $\rrc{FF,3}$ in the presence of $\Ioff$, where the reactions in the lower box in Fig.~\ref{fig:chemical_flipflop} guarantee (for appropriately chosen rate constants $\rrc{FF,4}$ and $\rrc{FF,5}$, cf.~Table~\ref{tab:rate_constants}) that $\Ioff$ is only available when \ac{STES} is no longer present.
Specifically, while \ac{STES} is active, $\Ioff$ is converted to $\Ion$ with a comparatively large rate constant $\rrc{FF,4}$, while $\Ion$ is converted to $\Ioff$ with comparatively small rate constant $\rrc{FF,5}$.
Hence, the presence of $\Ioff$ is an indicator that \ac{STES} is currently not active.

Despite technically not being a part of the flip-flop implementation, it is worth noting that at the beginning of the transmission of each pilot sequence, the external signal \ac{PMES}, discussed in more detail in the following subsection, catalyzes a reaction converting $\XtrueOn$ to $\XtrueOff$ molecules with rate constant $\rrc{FF,R}$ to ensure that the chemical flip-flop is aligned with the pilot symbol sequence (cf. Figure~\ref{fig:system_overview}-G).

\vspace*{-5mm}
\subsection{Switching between Pilot and Data Modes}\label{sec:timing:switching_pm_dm}
The external signals \ac{PMES} and \ac{DMES} switch the \ac{RX} between the pilot and data modes. To this end, the \ac{RX} utilizes $\Ptrain$ and $\Ptransmit$ molecules to keep track of its current operation mode, i.e., if $n_{\Ptrain}(t) \gg n_{\Ptransmit}(t)$, the \ac{RX} is in pilot mode at time $t$, whereas if $n_{\Ptrain}(t) \ll n_{\Ptransmit}(t)$, the \ac{RX} is in data mode.
\mbox{Figure~\ref{fig:system_overview}-C} shows how \ac{PMES} and \ac{DMES} convert $\Ptrain$ to $\Ptransmit$ and vice versa.
As indicated in Figures~\ref{fig:system_overview}-E and \ref{fig:system_overview}-F, the production of the helper molecules, as discussed in the previous subsection, and the generation/deletion of weight molecules in agreement with \eqref{eq:adaptive_learning_reaction} occur only in the pilot mode.
In particular, we note that the training reactions \eqref{eq:adaptive_learning_reaction} are extended in the context of the fully chemical \ac{RX} implementation to require the presence of both, $\Don$ and $\Ptrain$ molecules, as catalysts, cf.~Figure~\ref{fig:system_overview}-F.

\vspace*{-5mm}
\subsection{Chemical Cleanup}\label{sec:timing:cleanup}
\label{sec:timing_reset_phase}
As discussed in Section~\ref{sec:timing:chemical_switch}, the reset phase switch produces $\Ron$ molecules once a predefined amount of time has elapsed on the chemical stopwatch.
These $\Ron$ molecules catalyze several chemical cleanup tasks that reset the \ac{RX} before the start of the next symbol interval.
Specifically, the chemical stopwatch is reset by the $\Ron$-catalyzed conversion of $\Aon$ to $\Aoff$ molecules with rate constant $\rrc{dea}$ and the conversion of timer molecules $\Ton$ back to $\Toff$ with rate constant $\rrc{t,off}$ as shown in Figure~\ref{fig:system_overview}-A.
Furthermore, the $\Ron$ molecules also reset the detection switch by converting $\Don$ to $\Doff$ molecules with rate constant $\rrc{D,R}$.
Finally, as shown in Figure~\ref{fig:system_overview}-D, the $\Ron$ molecules remove any $\XhatCon$ and $\XhatCoff$ molecules with rate constants $\rrc{on,deg}$ and $\rrc{off,deg}$, respectively, in order to reset the counting mechanism of the chemical detector. 
The chemical cleanup is terminated by \ac{STES} at the beginning of the next symbol interval, since \ac{STES} catalyzes the conversion of all $\Ron$ to $\Roff$ molecules, cf. Figure~\ref{fig:system_overview}-B.

We will confirm in the next section that the proposed fully chemical \ac{RX} implementation indeed successfully integrates the various functional blocks of the proposed adaptive \ac{RX}.
Most notably, the various chemical building blocks introduced and discussed in this section together form a cellular \ac{RX} implementation that is not only capable of reliably detecting single transmitted symbols, but also to adapt to changing channel conditions over the course of multiple subsequent transmissions of data and pilot symbols.

\vspace*{-7mm}
\section{Simulation Results}
\label{sec:simulations}
In this section, we evaluate the performance of the \ac{RX} architecture proposed in this paper for different transmission scenarios.
First, we introduce the considered scenarios in Section~\ref{sec:simulations_setup}, before we confirm in Section~\ref{sec:basic_implementation_performance} that both the \ac{BM}-based detector and the adaptive detector are capable of approaching the optimum, i.e., \ac{MAP}, decision thresholds, even under very challenging transmission conditions.
Finally, in Section~\ref{sec:crn_based_implementation_results}, we evaluate the performance of the complete chemical \ac{RX} presented in Section~\ref{sec:timing} when employed in each of the considered transmission scenarios and compared to different theoretical baselines.

\vspace*{-3mm}
\subsection{Transmission Scenarios and Simulation Setup}
\label{sec:simulations_setup}
\vspace*{-2mm}
With respect to the channel model defined in \eqref{eq:system_model_concentrations}, we consider the following three transmission scenarios:
\begin{scenario}
    \item $\Deltac = \mu_{\Delta}$ constant, $N[l] = \mu_{N}$ constant for all symbol intervals $l$.\label{scenario:1}
    \item $\Deltac = \mu_{\Delta}$ constant, $N[l] \sim \textnormal{Lognormal}\Big(\tilde{\mu}_{N}, \tilde{\sigma}^{2}_{N}\Big)$ \ac{iid} for each $l$.\label{scenario:2}
    \item $\Deltac\sim \textnormal{Lognormal}\Big(\tilde{\mu}_{\Delta}, \tilde{\sigma}^{2}_{\Delta}\Big)$ \ac{iid}, $N[l]\sim \textnormal{Lognormal}\Big(\tilde{\mu}_{N}, \tilde{\sigma}^{2}_{N}\Big)$ \ac{iid} for each $l$.\label{scenario:3}
\end{scenario}
Here, $L \sim \textnormal{Lognormal}\left(\mu, \sigma^2\right)$ indicates that \ac{RV} $L$ is log-normal distributed, i.e., $\ln(L)$ is a Gaussian \ac{RV} with mean $\mu$ and variance $\sigma^2$.

\ref{scenario:1} corresponds to the default transmission model considered in most works on (static) diffusion-based \ac{MC}, while \ref{scenario:2} and \ref{scenario:3} introduce additional randomness in the concentrations of the background molecules and desired signaling molecules that can, for example, account for stochastic interference \cite{kuscu_adaptive_receiver_preprint} and \ac{TX} and/or \ac{RX} mobility \cite{ahmadzadeh_mobile_mc_stochastic_model}.
In terms of determining the optimal decision threshold for detection from observations of the received signal, the scenarios become more challenging from \ref{scenario:1} to \ref{scenario:3} as more sources of randomness are added in each scenario.

Throughout this section, $\tilde{\mu}_{\Delta}$,$\tilde{\sigma}^{2}_{\Delta}$, $\tilde{\mu}_{N}$, and $\tilde{\sigma}^{2}_{N}$ are chosen such that $\expectation{}{\Deltac} = \mu_{\Delta} = 10^{20}\uc$, $ \Var{\Deltac} = \sigma^{2}_{\Delta} = \frac{1}{4} \mu^2_{\Delta} =  \left(5 \times 10^{19}\uc\right)^2$, $\expectation{}{N} = \mu_{N}=1.5 \cdot 10^{19}\uc$, and $\Var{N} = \sigma^2_{N} = \frac{1}{4} \mu^2_{N} = \left(7.5 \times 10^{18} \uc\right)^2$.
We choose $ \mu_{\Deltac} = 10^{20}\uc$, as this is the order of magnitude of the peak molecule concentration which is expected if $N=1000$ molecules are released instantaneously from a point source in an unbounded 3-D environment, propagate by diffusion only, and are sampled at a distance of $d=0.75\,\si{\micro\metre}$ from the source (according to \cite[Eq. (6)]{jamali_channel_modeling}).

Now, to generate the receptor occupancy $\Y$ (or rather the number of bound receptors $\Nrb$) from the signaling molecule concentration $C[l]$ for each symbol interval $l$, we assume first that the \ac{RX} cell is equipped with $\nr = 30$ cell surface receptors, i.e., a number of receptors for which detection errors are still likely to occur even when using sophisticated detection methods \cite{kuscu_ml_detectors}.
Then, (i) realizations of $\Deltac$ and $N[l]$ according to one of the scenarios \ref{scenario:1}--\ref{scenario:3} as defined above are drawn, (ii) these realizations are substituted in \eqref{eq:system_model_concentrations} to obtain a realization $\cx$ of $C[l]$ which (iii) is then used to obtain the binding probability $\dfrac{\cx}{\cx+\rrc{-}/\rrc{+}}$ for a single receptor, cf.~\eqref{eq:receptor_binding_probability}.
Finally, (iv) a realization of $\Nrb$ is obtained by drawing from a Binomial distribution, as $\Nrb \sim \textnormal{Binom}\Big(\nr,\dfrac{\cx}{\cx+\rrc{-}/\rrc{+}}\Big)$.
For the binding and unbinding rates of the receptors, we use rate constants $\rrc{+}=2 \cdot 10^{19} \,\si{\meter^3\second^{-1}}$ and $\rrc{-}=20 \,\si{\second^{-1}}$, respectively, which were reported in \cite{kuscu_ml_detectors}.
The parameter values related to the \ac{CRN} implementation are provided in Section~\ref{sec:crn_based_implementation_results}.

\vspace*{-3mm}
\subsection{Threshold Learning}\label{sec:basic_implementation_performance}
\vspace*{-2mm}
In this section, we investigate whether the detectors proposed in Sections~\ref{sec:BMbasedRX} and \ref{sec:adaptiveRX} are able to approach the respective optimal detection thresholds for scenarios \ref{scenario:1}--\ref{scenario:3} from either offline training (BM-based detector) or pilot symbol transmissions (adaptive detector).
For the detector designs proposed in Sections~\ref{sec:BMbasedRX} and \ref{sec:adaptiveRX}, detection errors may happen (a) if the detectors implement non-optimal thresholds or (b) if the sample average of the surrogate \ac{RV} $\hat{X}$ as computed by the chemical detectors, cf.~\eqref{eq:system_model:sampling}, differs from the true probability $\Pr\left[\hat{X} | \vec{Y} \right]$.
In order to isolate the threshold learning performance of the respective detector, i.e., (a), from the randomness introduced by its chemical implementation, i.e., (b), we consider in this section simplified versions of the proposed detectors for which the detection decisions depend not on their proposed chemical realizations.
In the following, we first provide some details on how the separation between (a) and (b) is achieved in the simulations, before we present the results of the evaluation.

\subsubsection{Learning Simulation}\label{sec:simulations_procedure}

The numerical results presented in the following subsection were obtained based on 25 realizations of the \ac{BM}-based and the adaptive detector, respectively, where each realization was trained independently on $\nsamples$ random training samples / pilot symbols, if not stated explicitly otherwise.
The \acp{BER} for each detector were then obtained empirically based on the detection of $10^6$ data symbols with each of the respective detector realizations.

\paragraph{\ac{BM}-based Detection}
The \acp{BM} are initialized with all-zero biases $\thetavec$ and random weight matrices $\W$. 
Furthermore, each matrix $\W$ is constructed from another random matrix $\W'=\frac{1}{2}(\V+\V^\transpose) \in \mathbb{R}^{(\nr+1)\times(\nr+1)}$, where the entries of $\V$ are \ac{iid} Gaussian \acp{RV} with zero mean and variance $\frac{1}{\nr+1}$, by setting the diagonal elements of $\W'$ to zero.
Next, the \acp{BM} are trained based on a two-step procedure.
First, $\nsamples$ training samples are generated and utilized to estimate the first- and second-order moments of $\Z=\left[X \; \Y\right]^\transpose$ as $\expectation{}{\Z} \approx \frac{1}{\nsamples}\sum_{i=1}^{\nsamples} \z_i$ and $\expectation{}{\Z\Z^\transpose} \approx \frac{1}{\nsamples}\sum_{i=1}^{\nsamples} \z_i \z_i^{\transpose}$, respectively, where $\z_i$ denotes the $i$-th training sample.
Second, using these estimates, the \ac{BM} is trained for 100 steps.
To this end, the first- and second-order moments of the current \ac{BM} configuration are estimated in each \ac{BM} training step using the Gibbs sampling algorithm (see \cite{mackay2003information} for an overview) which is run for $5 \cdot 10^4$ steps after discarding the first 500 samples.
Then, as previously discussed in Section~\ref{sec:bm:background} and formally defined in \cite[Eq. (4)]{poole2017_chemical_BMs}, the weight matrix $\W$ and bias vector $\thetavec$ are updated accordingly.
Finally, during the evaluation phase, $\Pr[\hat{X}=1|\Y=\y]$ is estimated using $5 \cdot 10^2$ Gibbs samples.

\paragraph{Adaptive Detection}
The thresholds of the adaptive detectors are initialized with $\nwa[0]=0$.
Then, learning rule \eqref{eq:update_rule} is applied for each pilot symbol realization.
During the evaluation phase, \eqref{eq:adaptive_rx_theorem_2_Pr} is utilized directly in order to determine the detection probability $\Pr\left[\hat{X}=1|\Y\right]$ and the resulting {\em hypothetical \ac{BER}}, i.e., chemical decision making and chemical learning are both not considered here.
Furthermore, we compute the expected \acp{BER} after different numbers of pilot symbol transmissions and in steady state, respectively, as obtained from a theoretical analysis of \eqref{eq:update_rule} based on the theory of discrete-state Markov processes.
The corresponding analytical expressions are provided in \version{Appendices~\ref{sec:app:dtmc} and \ref{sec:app:steady_state}}{the appendix of the extended arxiv version of this paper \cite{bahe_crn_jp_arxiv}} and referred to as {\em deterministic baseline} ($\nsamples$ finite) and {\em asymptotic deterministic baseline} ($\nsamples \to \infty$) in the following, since they reflect the \acp{BER} that could ideally be expected on average if \eqref{eq:update_rule} was applied deterministically for many different receptor occupancy realizations (ensemble average).

\subsubsection{Learning Evaluation}
\begin{figure}[t]
    \vspace*{-18mm}
    \begin{minipage}{0.49\textwidth}
        \centering
        \includegraphics[width=3.4in]{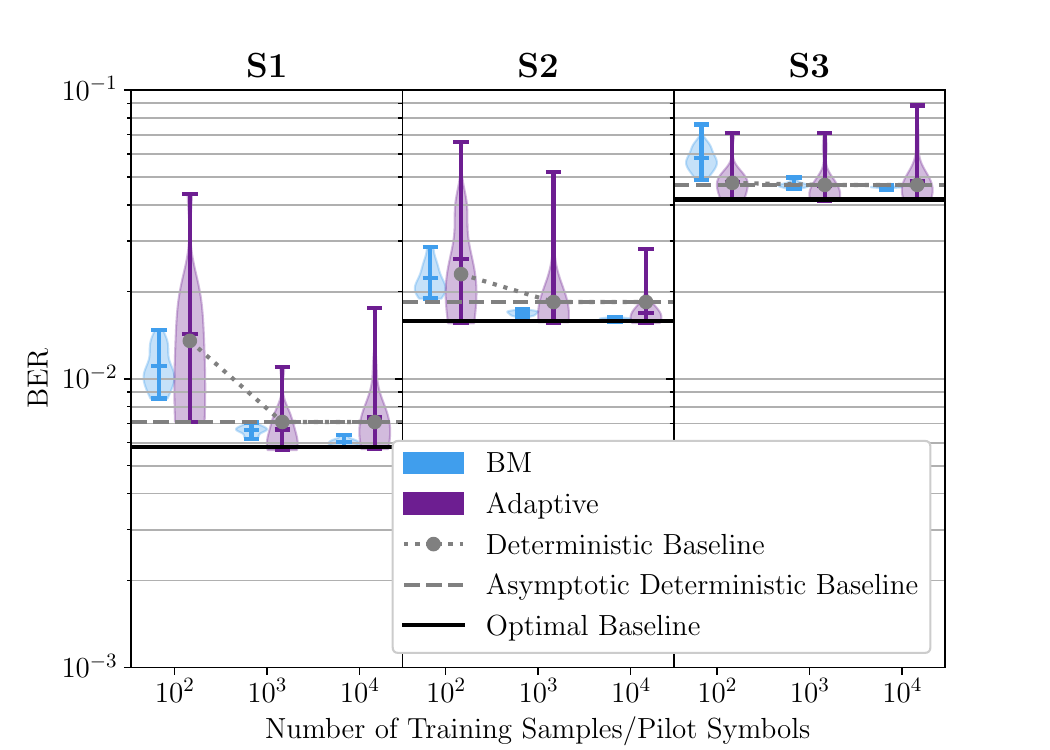}
    \end{minipage}
    \hspace*{-8mm}
    \begin{minipage}{0.50\textwidth}
        \vspace*{-5mm}
        \caption{\textbf{Detector comparison.} Hypothetical \acp{BER} after training the \ac{BM}-based (blue) and adaptive (purple) detectors with $\nsamples = 10^2$, $\nsamples = 10^3$, and $\nsamples = 10^4$ training samples/pilot symbols if suboptimal detection was only due to mismatched thresholds. The upper, center, and lower bar indicate the maximum, mean, and minimum \acp{BER} of the 25 detectors considered for each scenario. The deterministic and asymptotic deterministic baselines, as defined in the main text, are shown by the dotted and dashed grey lines, respectively. For reference, the \acp{BER} achievable by the respective \ac{MAP} detectors are shown as solid black lines.}
        \label{fig:in_silico}
    \end{minipage}
    \vspace*{-10mm}
\end{figure}

Figure~\ref{fig:in_silico} shows the hypothetical detection performance at different levels of training ($\nsamples = 10^2$, $\nsamples = 10^3$, and $\nsamples = 10^4$) for each of the three transmission scenarios defined in Section~\ref{sec:simulations_setup} and the two detectors proposed in Sections~\ref{sec:BMbasedRX} and \ref{sec:adaptiveRX}, respectively.
In addition, the \acp{BER} achievable by the respective \ac{MAP} detectors for each scenario, which constitute lower bounds on the detection error, are provided.
Finally, Figure~\ref{fig:in_silico} also shows the deterministic and the asymptotic deterministic baseline, respectively, as defined above.

From Figure~\ref{fig:in_silico}, we observe that the expected \acp{BER} of all detectors increase from Scenario \ref{scenario:1} to Scenario \ref{scenario:3}, which is expected due to the additional sources of randomness in each scenario.
Furthermore, Figure~\ref{fig:in_silico} reveals that in all considered scenarios the respective \acp{BER} decrease on average for both the \ac{BM}-based detector and the adaptive detector when $\nsamples$ is increased.
Furthermore, we observe that, as $\nsamples$ increases, the mean \acp{BER} obtainable with the \ac{BM}-based detector approach the performance of the corresponding \ac{MAP} detectors; also, the mean \acp{BER} of the adaptive \ac{RX} agree well with the theoretical analysis in the transient regime ($\nsamples$ finite), eventually approaching the theoretically predicted steady state.
Here, the offsets between the performance of the \ac{MAP} detectors and the asymptotically achievable performance of the adaptive detectors (asymptotic deterministic baseline) result from the fact that, as discussed in Section~\ref{sec:in_silico_learning}, the threshold update as defined for the adaptive detector in \eqref{eq:update_rule} is suboptimal.

We also observe from Fig.~\ref{fig:in_silico} that the \acp{BER} obtained with the different realizations of the adaptive detector are subject to stronger stochastic fluctuations as compared to the ones obtained from the \ac{BM}-based detectors.
In particular, the differences between the lowest and highest \ac{BER} of any simulated \ac{BM}-based detector, as indicated by the blue error bars in Fig.~\ref{fig:in_silico}, decrease for all considered transmission scenarios as $\nsamples$ increases.
The adaptive detectors, in contrast, exhibit comparatively stronger fluctuations of their \acp{BER} in all scenarios and even for the largest considered $\nsamples$, this variability does not vanish.
This observed sensitivity of the adaptive detectors towards random fluctuations in the learning process is again due to the choice of the {\em ad hoc} update rule \eqref{eq:update_rule}.
However, we will see later on that the end-to-end performance of the proposed fully chemical \ac{RX} is still competitive, while \eqref{eq:update_rule} facilitates a low complexity chemical implementation.

We conclude from the discussion in this section that the \ac{BM}-based detectors show in general superior performance and more robustness towards stochastic fluctuations as compared to the adaptive detectors as far as their ability to learn the respective optimal thresholds is concerned.
However, the training of the \ac{BM}-based detectors depends on the assumption that the channel statistics are, in some form, available offline and do not change over time.
In the next section, we will consider the realistic case as here this assumption is not fulfilled and focus on the adaptive \ac{RX}.

\vspace*{-5mm}
\subsection{CRN-based Implementation of the Adaptive Detector}\label{sec:crn_based_implementation_results}
\vspace*{-2mm}
In the remainder of this section, we study the performance of the fully chemical implementation of the adaptive \ac{RX} as described in detail in Section~\ref{sec:timing}.
First, the setup of the \ac{CRN} simulations is discussed in Section~\ref{sec:evaluation:adaptive_detector:simulation_setup}.
Then, the chemical realizations of the proposed timing mechanism and the pilot symbol-based threshold learning are studied in Sections~\ref{sec:evaluation:adaptive_detector:timing} and \ref{sec:evaluation:adaptive_detector:learning}, respectively.
Finally, the end-to-end performance of the proposed fully chemical \ac{RX} in terms of its expected \ac{BER} is studied in Section~\ref{sec:evaluation:adaptive_detector:ber}

\subsubsection{Simulation Setup}\label{sec:evaluation:adaptive_detector:simulation_setup}
All considered \acp{CRN} are simulated using the \ac{MNRM} algorithm \cite{thanh_MNRM}\footnote{The widely used Gillespie algorithm \cite{gillespie1977exact} is not directly applicable here, as its applicability hinges on the assumption of fixed rate constants. Hence, we employ the \ac{MNRM} which can account for time-varying rate constants such as those that occur in different reactions catalyzed by external signals considered in this paper (cf.~Section~\ref{sec:CRNprimer}).}.
Furthermore, in all simulations the initial molecule counts specified in Table~\ref{tab:initial_molecule_counts} and the normalized rate constants shown in Table~\ref{tab:rate_constants} are used unless noted otherwise.
As discussed in Section~\ref{sec:CRNprimer}, the normalized rate constants provided in Table~\ref{tab:rate_constants} can be scaled by a common factor without affecting the equilibrium state of the corresponding \acp{CRN} and this degree of freedom may by exploited in practice to realize their physical implementation.
We note that the reaction rate constants in Table~\ref{tab:rate_constants} were not optimized systematically, but rather chosen {\em ad hoc}.
While optimizing the (relative) rate constants in a systematic manner may possibly result in improved system performance, this is a non-trivial problem and left for future work.

Now, in order to simulate the considered \ac{MC} system, consecutive symbol intervals of normalized length $\Tsym=15$ each are considered.
For each symbol interval, one realization of the receptor occupancy $\Nrb$ is drawn according to the transmitted symbol and the current channel state.
\mbox{Furthermore, at} the beginning of each symbol interval, \ac{STES} and, if applicable, \ac{PMES} and \ac{DMES}, are applied for a duration of $\Ta=0.5$ in order to initialize either the pilot or the data mode.
As discussed in Section~\ref{sec:flip_flop}, whenever a pilot sequence is transmitted, it is of the form $\langle0,1,0,1,\dots \rangle$, i.e., alternating $0$'s and $1$'s, while the data symbols are selected independently and with equal probability in each symbol interval.

\begin{table}[!t]
    \centering
    \vspace{-10mm}
\begin{minipage}[t]{0.48\linewidth}\centering
    \caption{Initial Molecule Counts.}
    \vspace*{-2mm}
    \begin{tabular}{|c|c|c|c|c|c|}
        \hline
        $\Aon$ & $\Aoff$ & $\Ton$ & $\Toff$ & $\Ptrain$ & $\Ptransmit$ \\
        $0$ & $10$ & $0$ & $250$ & $10$ & $0$ \\
        \hline

        $\Don$ & $\Doff$ & $\Db$ & $\Ron$ & $\Roff$ & $\Rb$ \\
        $0$ & $50$ & $0$ & $50$ & $0$ & $0$ \\
        \hline

        $\Lon$ & $\Loff$ & $\Ion$ & $\Ioff$ & $\XtrueOn$ & $\XtrueOff$ \\
        $1$ & $0$ & $1$ & $0$ & $1$ & $0$ \\
        \hline

        \multicolumn{2}{|c|}{$\Won$} & $\XhatCon$ & $\XhatCoff$ & \multicolumn{2}{|c|}{$\Hmol$} \\
        \multicolumn{2}{|c|}{$0$} & $0$ & $0$ & \multicolumn{2}{|c|}{$0$}  \\
        \hline
    \end{tabular}
    \label{tab:initial_molecule_counts}
\end{minipage}
\hfill%
\begin{minipage}[t]{0.48\linewidth}\centering
    \caption{Normalized Reaction Rate Constants.}
    \vspace*{-2mm}
    \begin{tabular}{|c|c|c|c|c|}
    \hline
        $\rrc{a}$ & $\rrc{dea}$ & $\rrc{t,on}$ & $\rrc{t,off}$ & $\rrc{FF,1},\rrc{FF,2}$ \\
        $10^1$ & $10^{-1}$ & $10^{-2}$ & $2\cdot 10^{-2}$ & $5\cdot 10^{1}$ \\

        \hline
        $\rrc{D,i}$ & $\rrc{D,e}$ & $\rrc{D,R}$ & $\rrc{pilot}$ & $\rrc{FF,3},\rrc{FF,5}$ \\
        $10^0$ & $8 \cdot 10^{-1}$ & $5 \cdot 10^{-1}$ & $10^2$ & $10^{1}$ \\

        \hline
        $\rrc{R,i}$ & $\rrc{R,e}$ & $\rrc{R,R}$ & $\rrc{data}$ & $\rrc{FF,4}$ \\
        $10^0$ & $10^{-1}$ & $10^1$ & $10^2$ & $10^5$ \\

        \hline
        $\rrc{H}$ & $\rrc{H,deg}$ & $\rrc{l,1}$ & $\rrc{l,2}$ & $\rrc{FF,R}$ \\
        $10^{-1}$ & $5 \cdot 10^2$ & $1 \cdot 10^{-4}$ & $1 \cdot 10^{-4}$ & $10^1$ \\

        \hline
        $\rrc{on}$ & $\rrc{on,deg}$ & $\rrc{off}$ & $\rrc{off,deg}$ & $\rrc{deg}$ \\
        $5\cdot10^{0}$ & $10^{-1}$ & $5 \cdot 10^{0}$ & $10^{-1}$ & $5\cdot10^{4}$\\
        \hline
    \end{tabular}
    \label{tab:rate_constants}
\end{minipage}
\end{table}

\subsubsection{Chemical Timing Mechanism}\label{sec:evaluation:adaptive_detector:timing}
\begin{figure}[t]
    \vspace*{-16mm}
    \begin{minipage}{0.55\textwidth}
        \centering
        \includegraphics[width=3.4in]{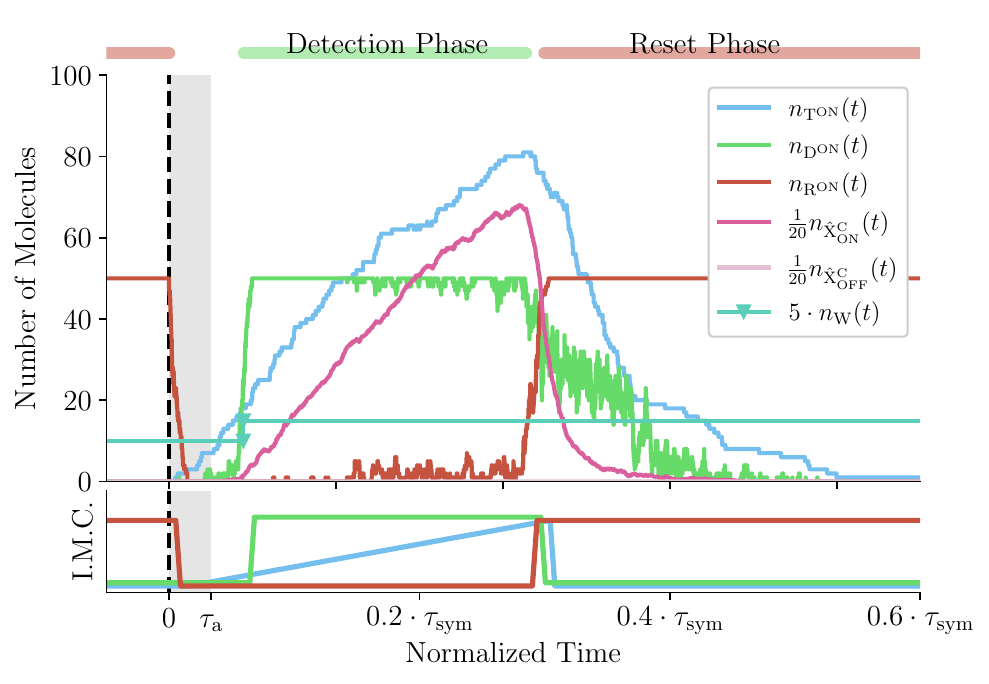}
        \vspace{-5mm}
        \caption{\textbf{Exemplary symbol interval.} \textit{Top}: After the symbol interval starts (dashed black line), \ac{STES} (grey shaded) generates $\Aon$ molecules (not shown) which in turn catalyze the production of timer molecules $\Ton$ (blue). After some time, first the detection phase (green $\Don$ molecules) and later the reset phase (red $\Ron$ molecules) start. In this interval $\xtrue=0$, but $\xhat=1$ (many pink $\XhatCon$ molecules are produced). Thus, $\nwa$ is adapted by generating an additional weight molecule $\Won$ (turquoise, rescaled for better visibility) is generated. The molecule counts of $\XhatCoff$ are barely visible since the corresponding curve (light pink) almost coincides with the $x$ axis. \textit{Bottom}: Idealized Molecule Counts (I.M.C.) for this symbol interval.}
        \label{fig:symbol_interval}
    \end{minipage}
    \hfill
    \begin{minipage}{0.44\textwidth}
        \centering
        \vspace*{0mm}
        \hspace*{-5mm}
        \includegraphics[width=3.4in]{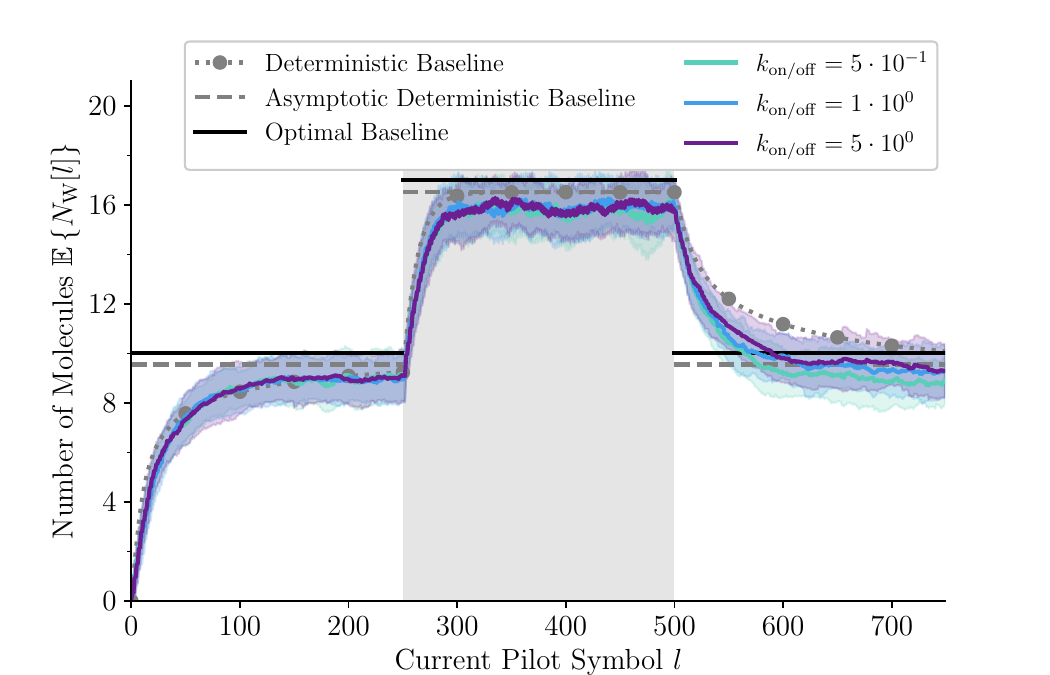}
        \vspace*{-10mm}
        \caption{
        \textbf{Time-variant channel.} Evolution of $\expectation{}{N_{\Won}[l]}$ of the chemically implemented adaptive detector for different rate constants $\kon$, $\koff$ (green, blue, purple). For reference, the deterministic, deterministic asymptotic, and optimal baseline (dotted grey, dashed grey, and solid black lines) are shown. The grey-shaded background indicates an increased background noise level.}
        \label{fig:Won_evolution}
    \end{minipage}
    \vspace*{-10mm}
\end{figure}

In order to validate the correct functioning of the chemical timing mechanism proposed in Section~\ref{sec:timing}, we study the time-dependent evolution of molecule counts within one symbol interval.
Specifically, Fig.~\ref{fig:symbol_interval} shows the evolution of molecule counts of the key species $\Ton$, $\Don$, $\Ron$, $\Won$, $\XhatCon$, and $\XhatCoff$ over the duration of one symbol interval after the transmission of pilot symbol $0$.
The bottom panel of Fig.~\ref{fig:symbol_interval} shows schematically how the molecule counts of the three molecule species representative of (i) the time elapsed since the beginning of the symbol interval ($\Ton$), (ii) the detection phase ($\Don$), and (iii) the reset phase ($\Ron$), should evolve ideally.
As depicted in Fig.~\ref{fig:symbol_interval}, the number of $\Ton$ molecules should ideally increase monotonically starting from the beginning of the symbol interval (dashed black line) triggering first the beginning of the detection phase (indicated by green bar on top) and later the beginning of the reset phase (red bars on top) which also ends the detection phase.
The detection phase is ideally characterized by an elevated level of $\Don$ molecules, while during the reset phase an increased level of $\Ron$ molecules is expected.

The top panel of Fig.~\ref{fig:symbol_interval} displays the temporal evolution of the actual molecule counts as obtained via computer simulation for one exemplary symbol interval.
First, we observe here that all three simulated molecule counts, i.e., $n_{\Ton}(t)$, $n_{\Don}(t)$, and $n_{\Ron}(t)$, show general trends that are in excellent agreement with the respective idealized molecule counts.
Specifically, $n_{\Ton}(t)$ increases monotonically starting from the beginning of the symbol interval until the end of the detection phase, while $n_{\Don}(t)$ and $n_{\Ron}(t)$ are significantly elevated during the detection and reset phases, respectively.
Comparing $n_{\Ton}(t)$, $n_{\Don}(t)$, and $n_{\Ron}(t)$ from the top panel with the corresponding idealized molecule counts in the bottom panel of Fig.~\ref{fig:symbol_interval}, we observe the impact of the randomness inherent in the chemical reactions which manifests itself in high frequency fluctuations in the simulated molecule counts.
Also, all chemical reactions occur at finite rates, as can be seen from the longer decay times of $n_{\Ton}(t)$ and $n_{\Don}(t)$ after the end of the detection phase compared to the idealized case.
However, despite these practical imperfections, Fig.~\ref{fig:symbol_interval} shows that $n_{\Don}(t)$ and $n_{\Ron}(t)$ correlate strongly with the detection and reset phase, respectively, validating the fundamental design of the proposed chemical timing mechanism.

Finally, the top panel in Fig.~\ref{fig:symbol_interval} illustrates also the simulated molecule counts of $\mol{W}$, $\XhatCon$, and $\XhatCoff$ in response to a transmitted pilot symbol $\xtrue=0$.
Furthermore, the decision threshold at the beginning of the interval, i.e., $n_{\mol{W}}$, equals $2$.
Now, we observe from Fig.~\ref{fig:symbol_interval} that many more $\XhatCon$ molecules than $\XhatCoff$ molecules are generated (since in the considered symbol interval the number of occupied receptors exceeds $2$).
Furthermore, the count of the latter ones is very close to $0$ most of the time.
Since $n_{\XhatCon} \gg n_{\XhatCoff}$ while $\xtrue=0$, a detection error occurs and, consequently, the chemical learning mechanism \eqref{eq:adaptive_learning_reaction} is triggered.
This ultimately results in the \mbox{increase in $n_{\mol{W}}$ observable in Fig.~\ref{fig:symbol_interval}.}

Our observations confirm that the chemical timing mechanism proposed in Section~\ref{sec:timing} operates as expected.
Furthermore, at any time during the considered symbol interval, relatively few molecules were used to perform computations while the number of molecules considered in order to execute comparably complex computations is commonly much larger in \ac{MC} systems.
Hence, the presented results confirm that the proposed chemical \ac{RX} architecture can in principle enable detection. Moreover, adaptation of the decision threshold to the previously unknown channel even at the low molecule counts that are practically feasible for cellular \acp{RX} is realized.
Since the adaptation part appears especially relevant for practical applications, we will study it in more depth in the following subsection.

\subsubsection{Adaptation to Time-variant Channel}\label{sec:evaluation:adaptive_detector:learning}
\label{sec:simulations_time_variant_channel}
In order to evaluate the ability of the proposed fully chemical adaptive \ac{RX} to adapt to a time-variant channel, a long sequence of $750$ pilot symbols is transmitted assuming Scenario~$\ref{scenario:1}$\footnote{In order to avoid desynchronization issues of the chemical flip-flop that may arise for unusually long pilot sequences, the chemical flip-flop is reset here every second symbol, i.e., each value '0' of $\xtrue$ is enforced by applying external signal \ac{PMES}.}.
Furthermore, after 250 symbols, $N[l]$ is increased by a factor of 5 and after 250 additional symbols reset back to its initial value, i.e., $N[l]=\mu_N$ for $0 \leq l < 250$, $N[l]=5 \mu_N$ for $250 \leq l < 500$, and $N[l]=\mu_N$ for $500 \leq l < 750$.

Fig.~\ref{fig:Won_evolution} illustrates how the expected number of weight molecules $n_{\mol{W}}$ changes over time during the considered pilot symbol transmission when averaged over different stochastic simulations of the proposed adaptive \ac{RX}.
For comparison, different baselines are provided in Fig.~\ref{fig:Won_evolution}.
One baseline, referred to as the {\em optimal baseline} in the following, corresponds to the respective \ac{MAP} thresholds for the different considered channel conditions.
The other two baselines correspond to the deterministic and the asymptotic deterministic baselines introduced in Section~\ref{sec:basic_implementation_performance}.

In Figure~\ref{fig:Won_evolution}, the expected number of weight molecules of the \ac{CRN}-based implementation of the adaptive \ac{RX} after different numbers of pilot symbol transmissions are shown for different rate constants $\rrc{on}$, $\rrc{off}$.
The thick lines indicate the observed mean values of $\Nwa$ after the transmission of each pilot symbol and the shaded regions extend over one standard deviation \mbox{above and below the mean}.

We first observe from Fig.~\ref{fig:Won_evolution} that the agreement of the proposed adaptive \ac{RX} with the deterministic baseline is very good.
This observation confirms that the performance losses introduced by the chemical implementation as compared to the deterministic realisation of the proposed algorithm are not severe.
Second, we observe that the chemical implementation eventually approaches the theoretical baseline closely, confirming that the proposed adaptive \ac{RX} is capable of learning the optimal decision threshold from pilot symbol transmissions.
As we observe from Fig.~\ref{fig:Won_evolution} this is true not only for the initial learning phase, but also after the channel conditions and hence the optimal decision threshold change at $l=250$ and $l=500$, respectively.
These observations confirm that the proposed fully chemical adaptive \ac{RX} successfully adapts to a changing environment by adjusting its decision threshold to the time-varying channel conditions based on the transmission of pilot symbols.
Finally, Fig.~\ref{fig:Won_evolution} confirms the robustness of the proposed chemical \ac{RX} architecture: The number of weight molecules does not depend too much on the exact values of $\kon$/$\koff$.

\subsubsection{Performance Analysis of Fully Chemical Adaptive RX}\label{sec:evaluation:adaptive_detector:ber}
\begin{figure}[t]
    \vspace*{-10mm}
    \begin{minipage}{0.58\textwidth}
        \centering
        \vspace*{-7mm}
        \includegraphics[width=1.0\linewidth]{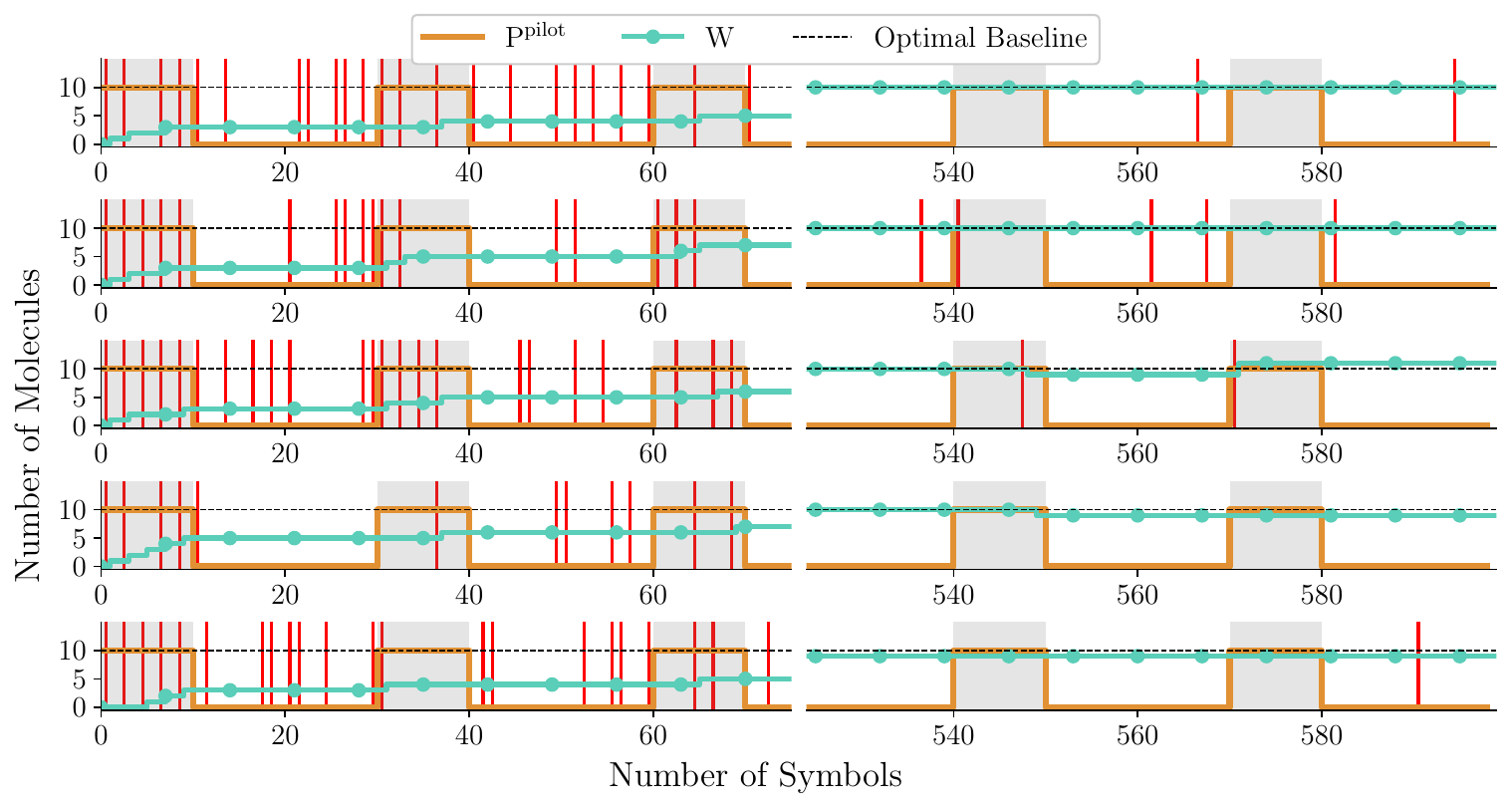}
        \vspace*{-10mm}
        \caption{\textbf{Alternating pilot and data phases.} The number of $\Won$ molecules (solid turquoise) is compared to the optimal decision threshold (dashed black) for five simulation runs with alternating pilot (grey shaded) and data (white background) intervals. The red vertical lines indicate bit errors in the respective symbol interval.}
        \label{fig:train_transmit}
    \end{minipage}
    \hfill
    \begin{minipage}{0.41\textwidth}
        \centering
        \includegraphics[width=1.0\linewidth]{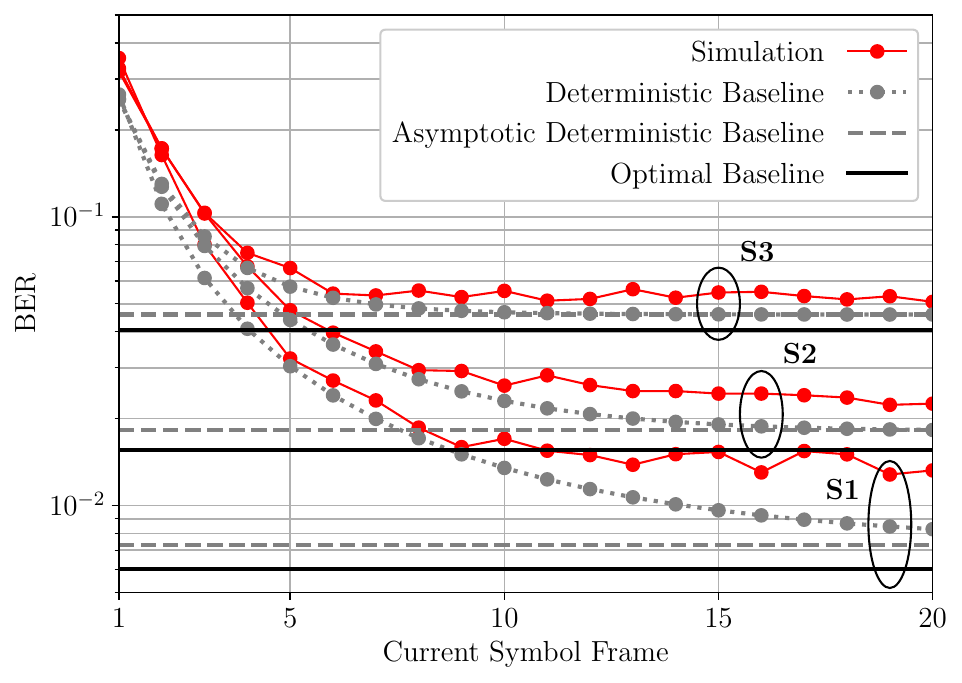}
        \vspace*{-10mm}
        \caption{\textbf{BER during data phases.} The evolution of the data-symbol \ac{BER} as obtained from the fully chemical implementation (red) is compared to the deterministic, asymptotic deterministic, and optimal baseline (dotted grey, dashed grey, and solid black lines) for all three scenarios.}
        \label{fig:evaluation:adaptive_detector:ber}
    \end{minipage}
    \vspace*{-10mm}
\end{figure}

To assess the performance of the fully chemical implementation of the proposed adaptive \ac{RX}, we consider the transmission of 20 frames, each comprising 10 pilot symbols followed by 20 data symbols, i.e., 30 symbols in total are transmitted per frame.
Furthermore, in order to define the \ac{BER} for the considered system, we first recall from Section~\ref{sec:adaptive_chemical_decisions} that the presence of $\XhatCoff$ and $\XhatCon$ is indicative of detection decisions $\xhat=0$ and $\xhat=1$, respectively, i.e., the detector output is not strictly binary.
Moreover, since the molecule counts of $\XhatCoff$ and $\XhatCon$ are subject to stochastic fluctuations, the detection decision is in general (for finite reaction rate constants) also probabilistic.
Accordingly, the \ac{BER} can be defined as the average probability to observe a molecule indicative of the wrong detection decision.
To this end, let $t_s$ and $t_e$ denote the start and end times of any symbol interval.
Then, for this interval, the probability of observing $\XhatCon$ relative to observing $\XhatCoff$, $\ponrel$, is defined as
\begin{equation}
    \ponrel = \frac{\int_{t_{\mathrm{s}}}^{t_\mathrm{e}} n_{\XhatCon}(t) \mathrm{d}t }{ \int_{t_{\mathrm{s}}}^{t_\mathrm{e}} n_{\XhatCon}(t) \mathrm{d}t + \int_{t_{\mathrm{s}}}^{t_\mathrm{e}} n_{\XhatCoff}(t) \mathrm{d}t },
\end{equation}
and the probability of wrong detection for this symbol interval equals $|x - \ponrel|$, where $x \in \{0,1\}$ denotes the transmitted symbol in this interval.
Averaging over all data symbols in one data frame and all realizations of the detector finally yields the \ac{BER}.

Fig.~\ref{fig:train_transmit} illustrates false detections along with the evolution of $\mol{W}$ molecules for $5$ different exemplary simulation runs for Scenario~$\ref{scenario:1}$ at the beginning and the end of the transmission of the $20$ frames, respectively.
Each horizontal panel in Fig.~\ref{fig:train_transmit} corresponds to $1$ simulation run.
For each simulation run and each transmitted symbol, Fig.~\ref{fig:train_transmit} displays a red vertical line whenever $|x - \ponrel| > 10^{-3}$.

First, we observe from Fig.~\ref{fig:train_transmit} that for all $5$ simulation runs errors are more likely to occur at the beginning than towards the end of the transmission.
Furthermore, we observe that during the transmission of pilot symbols (grey shaded symbol intervals) the number of $\Ptrain$ molecules (orange lines) is increased, enabling the adjustment of $n_{\mol{W}}$ whenever detection errors occur.
In contrast, when detection errors occur during information transmission, no such adaptations take place.
From the right part of Fig.~\ref{fig:train_transmit}, we observe that in all simulation runs, $n_{\mol{W}}$ is eventually very close to the theoretically optimal (\ac{MAP}) threshold, leading to very few detection errors.
In the following, we put this observation on a more quantitative basis.

Fig.~\ref{fig:evaluation:adaptive_detector:ber} shows the \acp{BER} as defined above obtained for stochastic simulations of the adaptive detector for the three different transmission scenarios introduced at the beginning of this section.
Specifically, for each transmitted frame, the simulated \acp{BER} in Fig.~\ref{fig:evaluation:adaptive_detector:ber} were obtained by averaging approximately 1000 independent realizations.
Furthermore, Fig.~\ref{fig:evaluation:adaptive_detector:ber} shows also the \acp{BER} obtained from the deterministic baselines and the asymptotic deterministic baselines as well as the expected \acp{BER} when the respective \ac{MAP} thresholds are used for detection.

First, we observe from Fig.~\ref{fig:evaluation:adaptive_detector:ber} that for all transmission scenarios the simulated \acp{BER} decrease over the course of the transmission.
As already observed for Scenario~\ref{scenario:1} in Fig.~\ref{fig:train_transmit}, this decay is due to the detectors adjusting their initial thresholds towards the optimum decision threshold for the respective scenario.
Furthermore, we observe that the decay of the simulated \acp{BER} follows closely the decay of the deterministic baseline, confirming that the error introduced by the proposed chemical implementation as compared to the performance of the proposed detector under deterministic conditions is not severe.
We also observe from Fig.~\ref{fig:evaluation:adaptive_detector:ber} that the performance loss in terms of the \ac{BER} introduced by the suboptimal threshold adaptation scheme \eqref{eq:update_rule} (as represented by the asymptotic deterministic baseline) is on the order of $10^{-3}$ for Scenarios \ref{scenario:1} and \ref{scenario:2}, and on the order of $10^{-2}$ for Scenario \ref{scenario:3}. This performance loss is clearly acceptable given the extremely low complexity of \eqref{eq:update_rule} and the versatility of the proposed detector.
Finally, we observe that the performance loss of the proposed fully chemical adaptive \ac{RX} as compared to the \ac{MAP} baseline after the transmission of only $20$ symbol frames is already on the order of $10^{-2}$ for all considered scenarios.

In summary, the results presented in this section confirm that the proposed fully chemical \ac{RX} does indeed deliver competitive performance in terms of the achieved \ac{BER} while at the same time successfully adapting to previously unknown channels.
This result is particularly satisfying since the proposed \ac{RX} realizes its complex functionality in a purely chemical fashion and the complexity of the involved chemical processes is relatively low as compared to existing approaches.
In light of these considerations, we believe that the proposed chemical \ac{RX} design will contribute towards the realization of collaboratively operating adaptive nanodevices in the context of the \ac{IoBNT}.

\vspace*{-2mm}
\section{Conclusion}\label{sec:conclusion}
\vspace*{-2mm}
In this paper, we presented a practically implementable architecture for cellular receivers based on \acp{CRN}. We introduced \ac{CRN}-based implementations of two detection mechanisms, one based on \acp{BM} that can be trained before deployment, and an adaptive one that can be trained using pilot symbols after deployment. Furthermore, we introduced a fully chemical realization of the proposed adaptive \ac{RX} including a \ac{CRN}-based timing mechanism to coordinate the different computation phases. We validated the proposed \ac{RX} with extensive stochastic simulations and provided a theoretical error analysis of the proposed adaptive training algorithm. 

Our results confirm that both proposed \ac{RX} designs learn (nearly) optimal decision thresholds when exposed to training data.
For the proposed adaptive \ac{RX} design, we showed that the proposed chemical timing mechanism is able to successfully schedule the different operation modes of the \ac{RX}.
Furthermore, we demonstrated that the adaptive \ac{RX} successfully adapts to changing channel states by adjusting its decision threshold.
Finally, we confirmed by extensive computer simulations that the proposed adaptive \ac{RX}, despite its relative simplicity and high versatility, delivers excellent performance when compared to the theoretical optimum; we were able to show that this last observation is even true for very challenging channel conditions in which both the background noise and the signaling molecule concentration change randomly from symbol interval to symbol interval.

In summary, we proposed an \ac{RX} architecture that is well-suited to narrow the implementation gap as we jointly design detection mechanisms and their chemical implementation. The practical relevance of our \ac{RX} is ensured by the explicit consideration of practical constraints like accounting for scenarios without mathematically tractable channel models, addressing the issue of synchronization, and aiming for a simple architecture with few chemical species and reactions.

We foresee four main future research directions. First, the external signals used for synchronization could be modeled in more detail, e.g., to account for background illumination. Second, the threshold adaptation for the adaptive \ac{RX} could be further improved, e.g., by allowing larger step sizes to speed up convergence. Third, the proposed timing mechanism might be used to facilitate the chemical implementation of other detectors, e.g., decision-feedback detectors for scenarios with non-negligible inter-symbol interference. Finally, it would be interesting to optimize the \ac{CRN} parameters subject to constraints of a chemical wetlab implementation and consider further chemical downstream processes that leverage the output of the proposed detector, e.g., to perform certain actions.

\version{
\appendices
\section{Training Rule: Markov Chain-based Error Analysis}
\label{sec:appendix_training_rule}
Here, we provide an error analysis for the steady-state behavior of the adaptive learning rule \eqref{eq:update_rule} based on a \ac{DTMC} formulation.

\subsection{DTMC Formulation}\label{sec:app:dtmc}
We assume that there can be up to $\nwtotal$ weight molecules in the system. Then, the current number of weight molecules $\nwa=n$ indicates the system state. For a \ac{DTMC}, the evolution over time is given by
\vspace*{-5mm}
\begin{equation}
    \pdistr[l+1] = \pdistr[l] \transmatrix,
\end{equation}
where $\pdistr[l]=\begin{bmatrix}\pi_0[l] & \dots & \pi_{\nwtotal}[l]\end{bmatrix} \in \mathbb{R}_{\geq0}^{(\nwtotal+1)}$ with $\sum_{i=0}^{\nwtotal} \pi_i[l]=1$. Here, the $n$-th entry $\pi_n[l]$ denotes the probability that there are $n$ active weight molecules in the system after the $l$-th training step. $\transmatrix$ is the transition matrix whose entry $P_{i,j}$ denotes the probability to transition from state $i$ to $j$ if the current state is state $i$. To construct $\transmatrix$ we assume that the pilot symbols are random and equiprobable. Here, we assume that $\kon=\koff$ and therefore, the decision threshold is given by $\nwa$ (cf. proof of Theorem~\ref{th:adaptive_map}). 
If the system is currently in state $n$, it moves out of this state if $x \neq \xhat$. The probability that we estimate $x=1$ as $\xhat=0$ and thus move from state $n$ to $n-1$ is
\begin{equation}
    P_{n,n-1} =\Pr[\Xhat=0,X=1|\Nwa=n] = \Pr[X=1] \cdot \sum_{i=0}^{n-1} \Pr[\Nrb=i|X=1].\label{eq:training_markov_chain_pnn-1}
\end{equation} 
Analogously, one can construct $P_{n,n+1}$. Because the only other option is to stay in the same state, $P_{n,n}=1-P_{n,n-1}-P_{n,n+1}$.

\subsection{Steady State Analysis}\label{sec:app:steady_state}
In order to assess the performance of the adaptive learning rule \eqref{eq:update_rule}, we are interested in the steady-state behavior, i.e., the performance once the training has converged. The considered \ac{DTMC} is a so-called birth-death process. Then, because the transition probability $P_{n-1,n}$ from state $n-1$ to state $n$ and the transition probability $P_{n,n-1}$ from state $n$ to state $n-1$, are known, the steady state distribution $\pdistrsteadystate$ can be computed recursively \cite{gallager_discrete_processes}. If we define the ratio $\rho_n = \frac{P_{n-1,n}}{P_{n,n-1}}$, the distributions can be computed as:
\begin{align}
    \pi^{\infty}_0 &= \frac{1}{1 + \sum_{i=1}^{\nwtotal} \prod_{j=1}^i \rho_j }, & n=0,\\
    \pi^{\infty}_n &= \rho_n \pi_{n-1}, & n>0.
\end{align}

The expected \ac{BER} for a given state distribution $\pdistr$ is given by
\begin{equation}
    \expectation{}{\BER} = \sum_{n=1}^{\nwtotal} \pi_n \cdot \BER_n = \sum_{n=1}^{\nw} \pi_n \cdot \left(1-P_{n,n}\right),\label{eq:mc_analysis_expected_BER}
\end{equation}
where $\BER_n$ denotes the \ac{BER} if $\Nwa=n$. By inserting $\pdistrsteadystate$, we can then obtain the expected \ac{BER} for the steady-state. Similarly, the expected \ac{BER} after $l$ pilot symbols can be computed by inserting $\pdistr[0] \transmatrix^l$ into \eqref{eq:mc_analysis_expected_BER}.
Analogously, the expected number of weight molecules can be computed as
\begin{equation}
    \expectation{}{N_{\Won}} = \sum_{n=1}^{\nwtotal} \pi_n \cdot n.
\end{equation}
}{}

\renewcommand{\baselinestretch}{1.21}
\bibliographystyle{IEEEtran}
\bibliography{bare_jrnl.bbl}

\end{document}